\documentclass[11pt]{article}
\usepackage{epsfig}
\usepackage{amssymb,amsmath,amsthm,amscd}
\usepackage{latexsym}
\usepackage[toc,page,title,titletoc,header]{appendix}
\usepackage{graphicx}
\usepackage{epsfig}
\usepackage{subfigure}
\usepackage{psfrag}
\usepackage{color}
\usepackage{subfigure}

\newtheorem{thm}{Theorem}[section]
\newtheorem{prop}[thm]{Proposition}
\newtheorem{cor}[thm]{Corollary}

\theoremstyle{definition}

\theoremstyle{remark}
\newtheorem{rem}{Remark}[section]

\def\eqref#1{(\ref{#1})}
\begin{document}
\begin{titlepage}
\title{An effect of large permanent charge: Decreasing flux to zero with increasing transmembrane potential to infinity}
\author{Liwei Zhang\footnote{School of Mathematical Sciences, Shanghai Jiao Tong University, 800 Dongchuan Road, Minhang District, Shanghai 200240, P. R. China ({\tt zhangliwei01@sjtu.edu.cn}).},\,
Bob Eisenberg\footnote{Department of Molecular Biophysics and Physiology, Rush Medical Center, 1759 Harrison St., Chicago, Illinois 60612 ({\tt beisenbe@rush.edu}).}, and 
Weishi Liu\footnote{Department of Mathematics, University of Kansas, 1460 Jayhawk Blvd.,   Lawrence, Kansas 66045 
({\tt wsliu@ku.edu}).}}
\date{\empty}
\end{titlepage}

\maketitle

\begin{abstract} In this work, we examine effects of large permanent charges on ionic flow through ion channels  based on a 
quasi-one dimensional Poisson-Nernst-Planck model. 
%This model  includes  the Poisson equation that determines the electrical potential from the charges present and   is consistent in that sense.  
It turns out large {\em positive} permanent charges inhibit the flux of cation as expected, but strikingly,  as the transmembrane electrochemical potential for anion increases in a particular way, the flux of anion decreases. The latter phenomenon was observed experimentally but the cause seemed to be unclear. The mechanisms  for these phenomena are examined with the help of the profiles of the ionic concentrations, electric fields and electrochemical potentials. 
The underlying reasons for the near zero flux of cation and for the decreasing flux of anion are shown to be different over different regions of the permanent charge.
Our model is   oversimplified. More structural detail and more correlations between ions can and should be included. But the basic finding   seems striking and important and deserving of further investigation.
\end{abstract}
  
\section{Introduction}
\setcounter{equation}{0}

Membranes define biological cells. They provide the barrier that separates, and defines the inside of a cell from the rest of the world.  Membranes are much more than just a barrier. They provide pathways for selected molecules to enter and leave cells. The barriers must  not be perfect or cells would  soon die from lack of energy or drown in their waste. Biological cells need energy to survive and that is provided (in almost all cases) by substances that must cross the membrane.

Substances cross membranes through proteins specialized for the task. For a very long time (\cite{Hod51,Uss49})
 these proteins have been separated into two classes, channels and transporters (\cite{Tos89}),
and studied in two traditions, one of electrophysiology (\cite{Hille01, ZT15}),
 the other of enzymology (\cite{Hille89,SL14}),
although the distinction between the two approaches was less than clearcut (\cite{Eis90}).
 
 Channels have been viewed fundamentally as charged `holes in the wall' (created by the membrane) that could open and close
 (\cite{SN95}) but, once open, the channel followed simple laws of electrodiffusion (\cite{Eis11}).

Transporters were viewed as more biological devices, involving conformation changes, coupling to energy sources (either ATP hydrolysis or the movement of other ions), with quantitative description much more difficult, particularly if the description was to be transferrable with parameters that were independent of conditions.

The enormous literature can be sampled in \cite{ASAG98, BBHS69, BL99,SL14,Tos89}.
 Structural biology has shown that transporters and channels have very similar structures (\cite{CMWSFRCBZ, LSRKGNK, SKSKKMN, WM17}).
 Biophysics has shown that the processes that open and close channels (`activation' and `inactivation') can be coupled to give properties rather like transporters. Physics has shown that the electric fields assumed to be constant in classical electrophysiology must depend on the distribution of charge in the channel and surrounding solutions (\cite{Eis96})
and so must change with experimental conditions and with location.

The detailed properties of open channels have not been compared with transporters in
the modern literature, as far as we know, certainly not using models that satisfy the physical requirement that potential profiles (i.e., electric fields) be computed from (and thus be consistent with) all the charges in the system.

Here we consider a simple model of a permanently open ion channel. (We leave gating for later consideration.) Most biologists imagine that if the driving force for electrodiffusion is increased--that is to say, if the gradient of electrochemical potential across the channel is increased in magnitude--the flux through the channel should increase. We show here that is not always the case. Consider a channel with large permanent charge and the flux of ions with the  opposite sign as the permanent charge (called  counter ions in the ion exchange literature (\cite{Hel62})
 or majority charge carriers in the semiconductor literature (\cite{Pie96, Sze81, VGK10}).
 The flux of ions with the opposite sign as the permanent charge in a channel can decrease dramatically as the driving force increases -- we term this phenomenon as {\em the declining phenomenon}. 
More precisely, if the concentration of the ion is held fixed on one side of the channel, and the concentration decreased on the other (`trans') side of the channel, the flux of the counter ions can decrease if the permanent charge density is large, as we show here. A depletion zone can form that prevents flow even though the driving force increases to large values.   It is worthwhile to emphasize  that if one increases the transmembrane electrochemical potential in a different manner, such as, by increasing the transmembrane electric potential or the concentration of the ion at one side of the channel, then one does not have the declining phenomenon (see Remark \ref{largemu}).

The decline of flux with {\em trans} concentration has been considered a particular, even defining proerties of transporters, involving conformation changes of state and other properties of proteins less well defined (physically) than electrodiffusion (\cite{KS59, Uss47}).
 Declining flux has been called exchange diffusion (in the early transport literature) and self-exchange more recently and is an example of obligatory exchange. Obligatory exchange is a wide spread, nearly universal property of the nearly eight hundred transporters known twenty years ago (\cite{ASAG98,GS97}) with many more known today (\cite{TK15}).
 Obligatory exchange is ascribed widely to changes in the structure of transport proteins, to conformation changes in the spatial distribution of the mass of the protein (\cite{FG86}).
 Obligatory exchange is often thought to be a special property of transporters not found in channels.
 
 The structure of many transporters is now known thanks to the remarkable advances of cryo-electron microscopy, recognized in the 2018 Nobel Prize. A transporter (of one amino acid sequence and thus of a perfectly homogeneous molecular type) exists in different states. Each state is said to have a different conformation meaning, in physical language, that the spatial distribution of mass is different in different states, and the distributions of the different states form disjoint sets, with no overlap. The movement of ions is not directly controlled or driven by the conformation of mass, however. Rather the distribution of mass produces a distribution of steric repulsion forces, and a spatial distribution of electrical forces (because the mass is associated with charge, mostly permanent charge of acid and base groups of the protein, but also significant polarization charge, as well). It is the conformation of these forces that determines the movement of ions. The spatial distribution of forces contributes to the potential of mean force reported in simulations of molecular dynamics.
 
This paper shows that channels with one spatial distribution of mass  can have properties of self-exchange (for  majority charge carrier counter ions) if the density of permanent charge is large. The spatial distribution of electrical forces can change and create a depletion zone that controls ion movement, while the spatial distribution of the mass of the protein does not change. The conformation governing current flow is the conformation of the electric field -- the depletion zone -- more than the conformation of mass, in the model considered here.
 The current flow of counter ions is much greater than the current flow of co-ions because there are many more counter ions than co-ions near the permanent charge. Transporters almost always allow much less current flow than channels.

It should be emphasized that the depletion zone considered here arises from the self-consistent solution of the Poisson-Nernst-Planck (PNP) equations of a specific model (large permanent charge,   counter ion transport) and that parameters are not adjusted in any way to create or modify the phenomena. This is in stark contrast to calculations of chemical kinetic models that involve many adjustable parameters, without clear physical meaning, and equations that do not conserve current (\cite{Eis14}).

Depletion zones play crucial roles in the behavior of nonlinear semiconductor devices although there they usually arise at locations in PN diodes where permanent charge (called doping in that literature) changes sign. Depletion zones of the type studied here are likely to occur in semiconductors but have received little attention because they have less dramatic effects than those in diodes (\cite{Pie96, Sze81, VGK10})
that follow drift diffusion and PNP equations rather like those of open ionic channels (\cite{Eis12,Eis96}).
 The possible role of depletion zones in channel function has been the source of speculation and experimental verification (\cite{Eis96-1,MVWMRE}). It is striking that depletion zones can change the conformation of the electric field and mimic the obligatory exchange traditionally thought to occur only in transporters. Depletion zones can create plastic electric fields whose change in shape dominate the flux through a channel of fixed structure.

Our model is of course oversimplified as are any models, or even simulations in apparent atomic detail, of condensed phases. More structural detail and more correlations between ions can and should be included. But the basic finding that large permanent charge can produce depletion zones and those regions can produce a decline of  counter ion flux when driving forces increase seems striking and important and deserving of further investigation.
 
 The rest of the paper is organized as follows. In Section \ref{setup}, we describe the quasi-one-dimensional Poisson-Nernst-Planck type model  and its dimensionless form for ionic flow. Assumptions are specified with two key assumptions: a dimensionless parameter $\varepsilon$ defined in (\ref{rescale}) is small and a dimensionless parameter $Q_0>0$ defined in (\ref{Q}) from the permanent charge is large.  In Section \ref{formulas},   approximation formulas  (in small $\varepsilon$ and $\nu=1/{Q_0}$) for fluxes are provided, which have a number of  non-trivial consequences. One of the apparent non-trivial consequences is that the leading order term $J_{10}$ of cation flux is zero, independent of the values of transmembrane electrochemical potential of the cation. The mechanism of this distinguished effect of large (positive) permanent charge is examined in details in Section \ref{Dyn4J10} with the help of the internal dynamics (profiles of the electric field, cation concentrations and electrochemical potential).  It turns out the mechanism for $J_{10}=0$ is different over different regions of permanent charge.
 In Section \ref{Dyn4declin}, a rather counter-intuitive {\em declining  phenomenon} -- increasing  of anion transmembrane electrochemical potential leads to  decreasing of anion flux -- is shown to be consistent with 
 our analytical result. Thus,   for the first time (to the best of our knowledge),  a mechanism of {\em large} permanent charge for such a phenomenon is revealed.  The mechanism is then examined   in details again with the help of the internal dynamics. We conclude this paper with a general remark in Section \ref{conclude}.

\section{Classical PNP with  large (positive) permanent charge}\label{setup}
\setcounter{equation}{0}
 
 \subsection{A quasi-one-dimensional PNP model}
Our study  is  based on a   quasi-one-dimensional PNP model for open channels, first proposed in \cite{NE} and, for a special case, rigorously justified   in \cite{LW}.
For a mixture of  $n$ ion species,
a quasi-one-dimensional    PNP  model   is
\begin{align}\begin{split}\label{ssPNP}
&\frac{1}{A(X)}\frac{d}{dX}\Big(\varepsilon_r(X)\varepsilon_0A(X)\frac{d\Phi}{dX}\Big)=-e_0\Big(\sum_{s=1}^nz_sC_s+{\cal Q}(X)\Big),\\
 & \frac{d{\cal J}_k}{dX}=0, \quad -{\cal J}_k=\frac{1}{k_BT}{\cal D}_k(X)A(X)C_k\frac{d\mu_k}{d X}, \quad
 k=1,2,\cdots, n
\end{split}
\end{align}
where $X\in [a_0,b_0]$ is the coordinate along the axis of the channel and baths, $A(X)$ is the
  cross-sectional area  of the channel at the location $X$, $e_0$ is the elementary charge, $\varepsilon_0$ is the vacuum permittivity, $\varepsilon_r(X)$ is the relative dielectric coefficient, ${\cal Q}(X)$ is the permanent charge density, $k_B$ is the Boltzmann constant, $T$ is the absolute temperature, $\Phi$ is the electrical potential,  and, for the $k$th ion species, $C_k$ is the concentration, $z_k$ is the valence, ${\cal D}_k(X)$ is the diffusion coefficient, $\mu_k$ is the electrochemical potential, and ${\cal J}_k$ is the flux density.

Equipped with system (\ref{ssPNP}),  a meaningful    boundary condition  for ionic flow through ion channels (see,   \cite{EL07} for a reasoning) is, for $k=1,2,\cdots, n$,
\begin{equation}
\Phi(a_0)={\cal V}, \ \ C_k(a_0)={\cal L}_k>0; \quad \Phi(b_0)=0,  \ \
C_k(b_0)={\cal R}_k>0. \label{ssBV}
\end{equation}

Mathematically, we will be interested in solutions of the boundary value problem (BVP) (\ref{ssPNP}) and (\ref{ssBV}). 
An important measurement for properties of ion channels is the {\em I-V (current-voltage) relation} where the current ${\cal I}$ depends on the transmembrane potential (voltage) ${\cal V}$ and is given by 
\[{\cal I}=\sum_{s=1}^nz_s{\cal J}_s({\cal V})\]
where ${\cal J}_k({\cal V})$'s are determined by the BVP (\ref{ssPNP}) and (\ref{ssBV}) for fixed ${\cal L}_k$'s and ${\cal R}_k$'s.  
Of course, the relations of individual fluxes ${\cal J}_k$'s to  ${\cal V}$ contain more information (\cite{JEL17})
but it is much harder to experimentally measure the individual fluxes ${\cal J}_k$'s.

\subsubsection{Electroneutrality boundary conditions}
In relation to typical experimental designs, the positions $X=a_0$ and $X=b_0$ are located in the baths separated by the channel. They are  the locations of two electrodes that are applied to control or drive the ionic flow through the ion channel.
Ideally, the experimental designs should not affect the intrinsic ionic flow properties so one would like to design the boundary conditions to meet the so-called electroneutrality 
\begin{align}\label{electroneutral}
\sum_{s=1}^nz_s{\cal L}_s=0=\sum_{s=1}^nz_s{\cal R}_s.
\end{align}
The reason for this is that, otherwise, there will be sharp boundary layers which cause significant changes  (large gradients) of the electric potential and concentrations near the boundaries so that a measurement of these values has non-trivial uncertainties. 
One clever design to remedy this potential problem is the ``four-electrode-design": two ``outer-electrodes'' in the baths far away from the ends of the ion channel to provide the driving force and two ``inner-electrodes'' in the baths near the ends of the ion channel to measure the electrical potential and the concentrations as the ``real'' boundary conditions for the ionic flow. At the inner electrodes locations, the electroneutrality conditions are reasonably satisfied, and hence, the electric potential and concentrations vary slowly and a measurement of these values would be robust. We point out that the geometric singular perturbation framework for PNP type  models developed in \cite{EL07,ELX15,JL12, LLYZ13,Liu05,  Liu09} can treat the case without electroneutrality assumption; in fact, the boundary layers can be determined by the boundary conditions directly. 

\subsubsection{Electrochemical potentials} The electrochemical potentials $\mu_k(X)$ consists of the ideal component $\mu_k^{id}$ and the excess component $\mu_k^{ex}$ where
the ideal component 
\begin{align}\label{idealComp}
\mu_k^{id}(X)=z_ke_0\Phi(X)+k_BT\ln \frac{C_k(X)}{C_0}
\end{align}
is the point-charge contribution where $C_0$ is a characteristic concentration, and the excess component $\mu_k^{ex}(x)$ accounts for ion size effects.  As explained above, although not totally physical for ion channel problems in general, we will consider only the ideal component in this work  that can act as    guidance  for further studies of more accurate models with excess component.  

\subsubsection{Permanent charges and channel geometry} The permanent charge ${\cal Q}(X)$ is a  simplified mathematical   model for ion channel (protein) structure.  It is determined by the spatial distribution of amino acids in the channel wall, the acid
 (negative) and base (positive) side chains, more than anything else (\cite{Eis96}).
  We will assume ${\cal Q}(X)$ is known and take an oversimplified description to capture some essence of its effects. For this paper, we take it to be as in (\cite{EL07}), for some $a_0<A<B<b_0$,
\begin{align}\label{calQ}
{\cal Q}(X)=\left\{\begin{array}{cc}
0, & X\in [a_0,A)\cup (B,b_0]\\
2{\cal Q}_0, &X\in (A,B).
\end{array}\right.
\end{align}
We will be interested in the case where $|{\cal Q}_0|$ is large relative to ${\cal L}_k$'s and ${\cal R}_k$'s.

 The cross-sectional area $A(X)$ typically  has the property that $A(X)$ is much smaller for $X\in (A,B)$ (the neck region) than that for $X\not\in (A,B)$. It is interesting to note that, in \cite{JLZ15}, the authors showed that the neck of the channel should be ``narrow" (small $A(X)$ for $X\in (A,B)$) and ``short" (small $B-A$) to optimize an effect of permanent charge.
 
\subsubsection{Dielectric coefficient and diffusion coefficients} We assume that 
  \begin{align}\label{diffusionCoeff}
 \varepsilon(X)=\varepsilon_r\,\mbox{ is a constant,  and }\; {\cal D}_k(X)= {\cal D}(X){\cal D}_k
  \end{align} 
  for some dimensionless  function ${\cal D}(X)$  (same for all $k$) and dimensional constant ${\cal D}_k$.
 
Note that the assumption ${\cal D}_k(X)= {\cal D}(X){\cal D}_k$  is equivalent to the statement that ${\cal D}_k(X)/{{\cal D}_j(X)}$ is a constant for $k\neq j$.   Roughly speaking, the assumption says that, 
as the environment varies from location to location, its influences on the two  diffusion coefficients ${\cal D}_k(X)$ and ${\cal D}_j(X)$ at the same location $x$ are   the same; that is, the two diffusion coefficients  vary from one common environment to another common environment in a way so that their ratio is independent of locations.     
This is   not a  justification of this assumption but only an explanation of what it reflects.  

\subsubsection{Main assumptions} In the sequel, we   assume {\em the boundary electroneutrality condition in (\ref{electroneutral}), the form of the permanent charge in (\ref{calQ}), constant dielectric coefficient and diffusion coefficient property in (\ref{diffusionCoeff}),
 and the electrochemical potential is ideal $\mu_k=\mu_k^{id}$ in (\ref{idealComp}).} There are two key assumptions for this work to be discussed in terms of dimensionless variables below. 
 
\subsection{Dimensionless form of the quasi-one-dimensional PNP model}
The following rescaling (see \cite{Gil99}) or its variations  have been widely used for convenience of mathematical analysis.   \medskip
 
  Let $C_0$ be a characteristic concentration of the ion solution. 
 We now make a dimensionless  re-scaling of the  variables in  system (\ref{ssPNP}) as follows.
\begin{align}\label{rescale}\begin{split}
&\varepsilon^2=\frac{\varepsilon_r\varepsilon_0k_BT}{e_0^2(b_0-a_0)^2C_0},\; x=\frac{X-a_0}{b_0-a_0},\;  h(x)=\frac{A(X)}{(b_0-a_0)^2},
\;  Q(x)=\frac{{\cal Q}(X)}{C_0},  \\
&D(x)={\cal D}(X),\; \phi(x)=\frac{e_0}{k_BT}\Phi(X), \; c_k(x)=\frac{C_k(X)}{C_0}, \;  
 J_k=\frac{{\cal J}_k}{(b_0-a_0)C_0   {\cal D}_k}. 
\end{split}
\end{align}
 
 The dimensionless quantity $Q(x)$ from the permanent charge ${\cal Q}$ in (\ref{calQ}) becomes
\begin{align}\label{Q}
Q(x)=\left\{\begin{array}{cc}
0, & x\in [0,a)\cup (b,1]\\
2Q_0, &x\in (a,b),
\end{array}\right.
\end{align}
where  
\[0<a=\frac{A-a_0}{b_0-a_0}<b=\frac{B-a_0}{b_0-a_0}<1,\]
and, in terms of the dimensionless quantities, the subinterval $(a,b)\subset (0,1)$ corresponds to the neck region $[A,B]$.

 \noindent
{\bf Key assumptions.} {\em The   parameter $\varepsilon$ is small and the parameter $Q_0>0$ is large.} 
 
 The case where $Q_0<0$ with $|Q_0|$ large can be treated in the same way.
 The   assumption of smallness of the dimensionless parameter $\varepsilon$  is reasonable particularly because we are interested in the limit of large $|Q_0|$. Thus, we can   take large but fixed characteristic concentration $C_0$ in the definition of $\varepsilon$ (\ref{rescale}). 
  For example, if $b_0-a_0=25nm$ and $C_0=10 M$, then the dimensionless parameter $\varepsilon\approx 10^{-3}$ (\cite{EL17}). The mathematical consequence of   smallness of $\varepsilon$  is that the BVP  (\ref{PNP}) and  (\ref{BVO}) can be treated as  a {\em singularly perturbed problem}.  
A general geometric framework for analyzing the singularly perturbed BVP of PNP type systems has been developed in \cite{EL07, Liu05, Liu09, LX15} for classical PNP systems and in \cite{JL12, LLYZ13, LTZ12} for PNP systems with finite ion sizes.   

 \medskip
  
  In terms of the new variables in (\ref{rescale}),   the  BVP (\ref{ssPNP}) and (\ref{ssBV}) becomes 
     \begin{align}\label{PNP}\begin{split}
&\frac{\varepsilon^2}{h(x)}\frac{d}{dx}\left(h(x)\frac{d\phi}{dx}\right)=-\sum_{s=1}^nz_s
c_s -Q(x),\\
&\frac{d J_k}{dx}=0, \quad -  J_k=\frac{1}{k_BT}D(x)h(x)c_k\frac{d \mu_k}{d x},
\end{split}
\end{align}
with  boundary conditions at $x=0$ and $x=1$
\begin{align}\label{BVO}\begin{split}
\phi(0)=&V,\; c_k(0)=L_k; \;
 \phi(1)=0,\; c_k(1)=R_k,
\end{split}
\end{align}
where 
\[V:=\frac{e_0}{k_BT}{\cal V},\quad L_k:=\frac{{\cal L}_k}{C_0},\quad R_k:=\frac{{\cal R}_k}{C_0}.\]
 
In this work, we  will consider the  BVP  (\ref{PNP}) and  (\ref{BVO}) for ionic mixtures with one cation of valence $z_1=1$ and an anion of valence $z_{2}=-1$. We will be interested in properties of ionic flow for  large $|Q_0|$.
 It turns out  
 \[\alpha=\frac{H(a)}{H(1)}\;\mbox{ and }\; \beta=\frac{H(b)}{H(1)}\;\;
\mbox{ where }\; H(x)=\int_0^x\frac{1}{D(s)h(s)}ds\]
are the key parameters that characterize the effect of channel geometry and location of permanent charge.
 The physical meanings of these parameters could be seen from the special case when $h(x)=h_0$ and $D(x)=D_0$ are constants. In this case, 
\[H(x)=\frac{x}{D_0h_0},\]
which is {\em proportional} to the (scaled) length $x$ of the region over $[0,x]$ of the channel (conductive material), {\em inversely proportional} to the (scaled) cross-sectional area $h_0$ of the channel and to the (scaled) diffusion coefficient or electric conductivity $D_0$; that is, $H(x)$ is the {\em resistance} of the portion of the channel over $[0,x]$.

\section{Approximations of fluxes \cite{ZL17}.}\label{formulas}
\setcounter{equation}{0}

We now recall the results on approximations of $(\phi, c_1,c_2,J_1,J_2)$ and $\mu_k$'s from \cite{ZL17}
for the case where $z_1=1$ and $z_2=-1$ with $L_1=L_2=L$ and $R_1=R_2=R$. 
 
Based on the   assumptions   that   $\varepsilon>0$ is SMALL and $Q_0>0$ is LARGE, one has expansions of fluxes in $\varepsilon$ and in $\nu=1/{Q_0}$ near $\varepsilon=0$ and $\nu=0$.   
  
  First, one expands the fluxes in $\varepsilon$ as
  \[J_1(\varepsilon;\nu)=J_1(\nu)+O(\varepsilon)\;\mbox{ and }\; J_2(\varepsilon)=J_2(\nu)+O(\varepsilon),\]
 where $J_k(\nu)$, depending also on the parameters $(V, L, R, H(1), \alpha,\beta)$,
 are the zeroth order in $\varepsilon$ terms of the fluxes. 
 Then, one  expands  $J_k(\nu)$ about   $\nu$ as 
\begin{align*} 
J_1(\nu)=J_{10}+J_{11}\nu+O(\nu^2)\;\mbox{ and }\; J_2(\nu)=J_{20}+J_{21}\nu+O(\nu^2).
\end{align*}
Thus, 
\begin{align}\label{expinnu}
J_1(\varepsilon;\nu)=J_{10}+J_{11}\nu+O(\nu^2,\varepsilon)\;\mbox{ and }\; J_2(\varepsilon;\nu)=J_{20}+J_{21}\nu+O(\nu^2,\varepsilon).
\end{align}
Note that $J_{10}$ and $J_{20}$, depending on $(V, L, R, H(1),\alpha,\beta)$, contain the leading order effect of the {\em small} $\nu$ (or equivalently, {\em large} $Q_0$). 

The following result is established in \cite{ZL17}.
  \begin{prop}\label{expJs} One has 
\begin{align}\label{ejJ}\begin{split}
J_{10}&=0,\\
J_{11}&=\frac{1}{2H(1)(\beta-\alpha)}\left(\frac{(1-\beta)L+\alpha R}{(1-\beta)\sqrt{e^{V}L}+\alpha\sqrt{R}}\right)^2(e^{V}L-R); \\
J_{20}&=\frac{2\sqrt{LR}}{H(1)}\frac{1}{(1-\beta)\sqrt{L}+\alpha\sqrt{e^{-V}R}}(\sqrt{e^{-V}L}-\sqrt{R}),\\
J_{21}&=\frac{(1-\beta)L+\alpha R}{2H(1)(\beta-\alpha)(\sqrt{e^{-V}L}-\sqrt{R})((1-\beta)\sqrt{L}+\alpha\sqrt{e^{-V}R})^3} \\
&\quad\left\{-2(\beta-\alpha)^2(\sqrt{e^{-V}L}-\sqrt{R})^2\sqrt{e^{-V}}LR\right.\\
&\quad\left.+((1-\beta)L+\alpha R)(L-Re^{-V})\left[\frac{(1-\beta)L+\alpha R}{2}\sqrt{e^{-V}}\ln{\frac{L}{e^VR}}\right.\right.\\
&\qquad\qquad\left.\left.-((1-\beta)\sqrt{L}+\alpha\sqrt{e^{-V}R})(\sqrt{e^{-V}L}-\sqrt{R})\right]\right\}.
\end{split}
\end{align}
\end{prop}
  
A distinct implication of $J_{10}=0$ in   Proposition \ref{expJs}  is that large (positive) permanent charge $Q_0$ (or small  $\nu=1/{Q_0}$) inhibits the cation flux.    We will provide a detailed discussion in Section \ref{Dyn4J10} on what happens to the  internal dynamics that is consistent with this conclusion.

To this end, we recall another immediate consequence of (\ref{ejJ})  (see \cite{ZL17} for more).

\begin{cor}\label{Saturate} [Current  Saturation] For large permanent charge  $Q_0$ (small ${\nu}=1/{Q_0}$) and to the leading order terms in $\nu$, each individual fluxes $J_k$'s, and hence, the total  current $I$ saturate as $|V|\to \infty$; more precisely, one has
\begin{align}\label{saturation}\begin{split}
\lim_{V\to +\infty}  J_{20}=&-\frac{1}{1-\beta}\frac{ 2R}{H(1)}, \quad \lim_{V\to -\infty}  J_{20}=\frac{1}{\alpha}\frac{2L}{H(1)},\\
\lim_{V\to +\infty} J_{11}=&-\lim_{V\to +\infty}J_{21}=\frac{1}{(1-\beta)^2}\frac{((1-\beta)L+\alpha R)^2}{2H(1)(\beta-\alpha) },\\
\lim_{V\to -\infty} J_{11}=&- \lim_{V\to -\infty}J_{21}=-\frac{1}{\alpha^2}\frac{((1-\beta)L+\alpha R)^2}{2H(1)(\beta-\alpha)}.
\end{split}
\end{align}
\end{cor}
\begin{proof} All the above (finite) limits can be derived from   (\ref{ejJ})  easily.    
\end{proof}
 
Note this is NOT the case if the permanent charge is small (\cite{JLZ15}).

%\newpage
\section{Internal dynamics for $J_{10}=0$.}\label{Dyn4J10}
\setcounter{equation}{0}

It follows from   the Nernst-Planck equation in (\ref{PNP}) that
\[J_1\int_0^1\frac{k_BT}{ D_1(x)h(x) c_1(x)}dx=\mu_1(0)-\mu_1(1).\]
Thus, $J_1$ has the same sign as that of $\mu_1(0)-\mu_1(1)$; in particular, if $\mu_1(0)-\mu_1(1)\neq 0$, then $J_1\neq 0$.

 For the setup of this paper,  there are three regions of permanent charge $Q(x)$: $Q(x)=0$ for $x\in [0,a)$ and $x\in (b,1]$ and $Q(x)=2Q_0$ for $x\in [a,b]$ with {\em large} $Q_0$ or {\em small} $\nu$ with ${\nu}=1/{Q_0}$. A major consequence of large $Q_0$ is, from Proposition \ref{expJs}, up to the leading order in $\nu=1/{Q_0}$ and in $\varepsilon$, the flux of cation is $J_{10}=0$, {\em even if  the transmembrane potential $\mu_1(0)-\mu_1(1)\neq 0$}.
We will  reveal the internal dynamics   that lead to this conclusion. To do so, we
 will discuss what happens over each subinterval     based on the approximated (of zeroth order in $\varepsilon$) functions of profiles. Let   
\begin{align}\label{leadv}
(\phi(x;\varepsilon,\nu),c_k(x;\varepsilon,\nu), J_k(\varepsilon,\nu))=(\phi(x;\nu),c_k(x;\nu) , J_k(\nu))+O(\varepsilon)
\end{align}
be the   solution of the BVP (\ref{PNP}) and (\ref{BVO}). For   $\nu>0$ small, one has the following expansions
\begin{align*}
J_1(\nu)=&J_{10}+J_{11}\nu+O(\nu^2), 
  \end{align*}
 where $J_{10}=0$ and $J_{11}$  are given in (\ref{ejJ}). 
We will also provide figures for the profiles of  rescaled electrical potential $\phi(x;\nu)$, cation concentration $c_1(x;\nu)$ and electrochemical potential    $\mu_1(x;\nu)$ of cation, respectively. The parameter values used are   
\begin{align}\label{parameters}\begin{split}
e_0=&1.6022 \times 10^{-19}\ (C),\ k_B=1.3806\times 10^{-23}\ (JK^{-1}),\ T=273.16\ (K),\\
\mathcal{V}=& 0.01\ (V), \ {\cal L}=0.5\ (M), \;{\cal R}=10^{-5}\ (M),  \; {\cal Q}_0=10^8\ (M),\; C_0=1\ (M),\\
  a_0=&0, \; b_0=10\ (nm),\; A=1/3\ (nm), \; B=1/2\ (nm).
 \end{split}
 \end{align}
In terms of the dimensionless parameters,  
 $$V \approx 0.425,\; L=0.5,\; R=10^{-5}\approx 0,\; \;a=\frac{1}{3},\;b=\frac{1}{2},\; \ \nu=10^{-8}\approx 0,$$
 The scaled cross-sectional area $h(x)$ is taken to be
\begin{align*}
h(x)=\left\{\begin{array}{ccc} \pi(-x+r_0+a)^2, & x\in(0,a)\\ 
\pi r_0^2, & x\in(a,b)\\ 
\pi(x+r_0-b)^2, & x\in(b,1).
\end{array}\right.
\end{align*}
The same setup will be used  for figures in Section \ref{Dyn4declin}.

Notice that the whole interval is $(0,1)$, and it is divided into three subintervals $(0,a), (a,b)$ and $(b,1)$. Permanent charge is located over $(a,b)$. The radius of the neck of ion channel is $c$.  In the figures, we let $r_0=0.5$. We make plots of each function in interval $(0,1)$, also in every subintervals.  

     \subsection{Internal dynamics over the interval $(0,a)$}
   The leading order terms of $(\phi, c_1)$ in (\ref{leadv}) are derived in \cite{ZL17}. One has
      \begin{prop}\label{expProfile11} For $x\in (0,a)$, 
  \begin{align*}
  \phi(x;\nu)=&\phi_0(x)+\phi_1(x)\nu+O(\nu^2) \mbox{ with }\\
  \phi_0(x)=& V- \ln\Big(1- \frac{J_{20}}{2L}H(x)\Big),\quad
 \phi_1(x)=-\frac{c_{11}(x)}{c_{10}(x)}+\frac{2J_{11}}{J_{20}}\ln{\frac{c_{10}(x)}{L}};\\
c_{10}(x)=& L- \frac{J_{20}}{2}H(x),\quad
 c_{11}(x)= -\frac{J_{11}+J_{21}}{2}H(x);\\
   \mu_{10}(x)=&\mu_{10}(0)=V+\ln{L}, \quad \mu_{11}(x)=2\frac{J_{11}}{J_{20}}\ln{\frac{c_{10}(x)}{L}}.
 \end{align*}
 \end{prop}

 Figure 1   shows the profiles of $c_1(x;\nu)$, $\phi(x;\nu)$ and $\mu_1(x;\nu)$ over the interval $(0,a)$.

  \begin{figure}[htbp]\label{c1one}
\centering
\begin{minipage}[t]{0.4\linewidth}
 \includegraphics[width=2.5in]{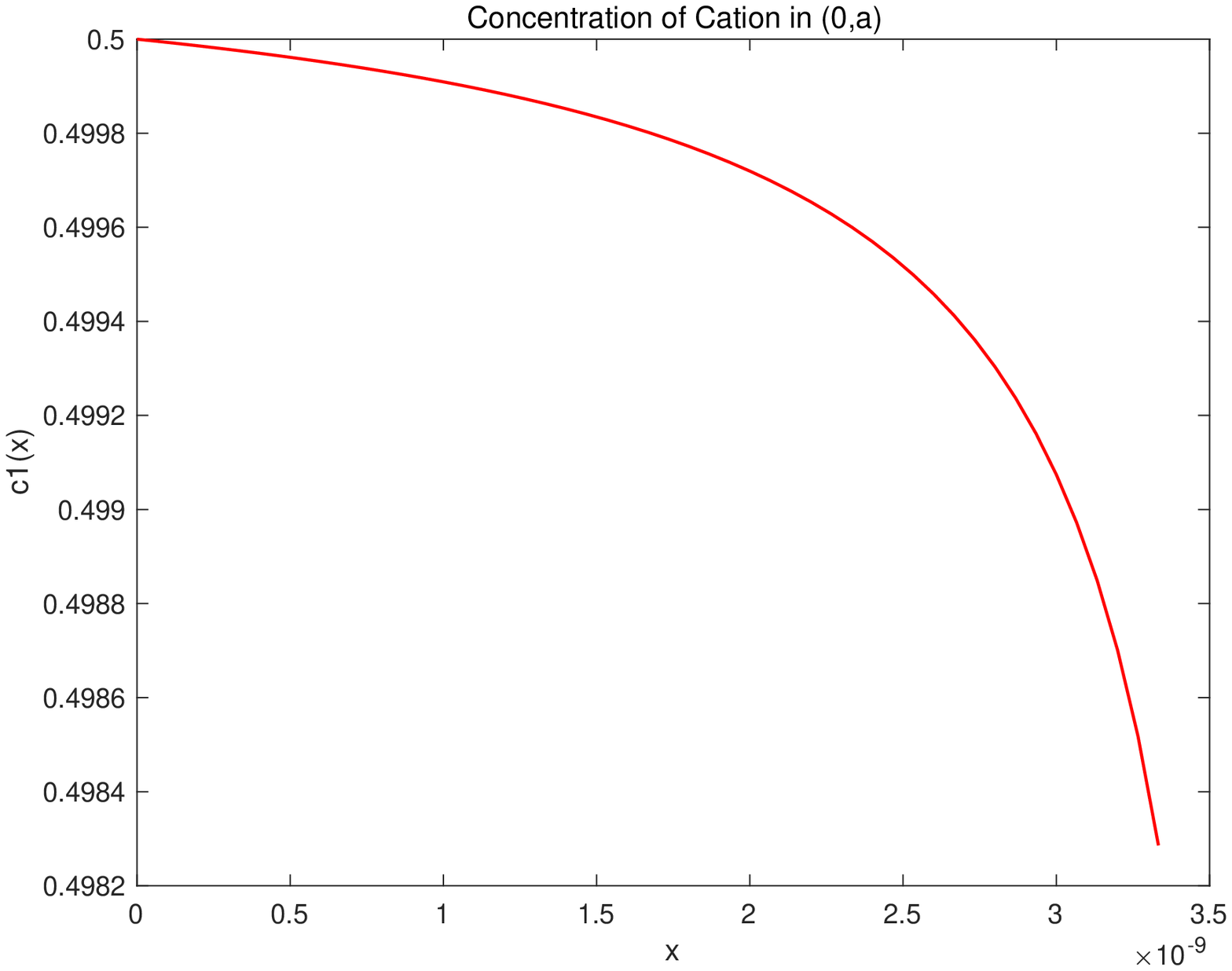}
 \end{minipage}
  \qquad
 \begin{minipage}[t]{0.4\linewidth}
 \includegraphics[width=2.5in]{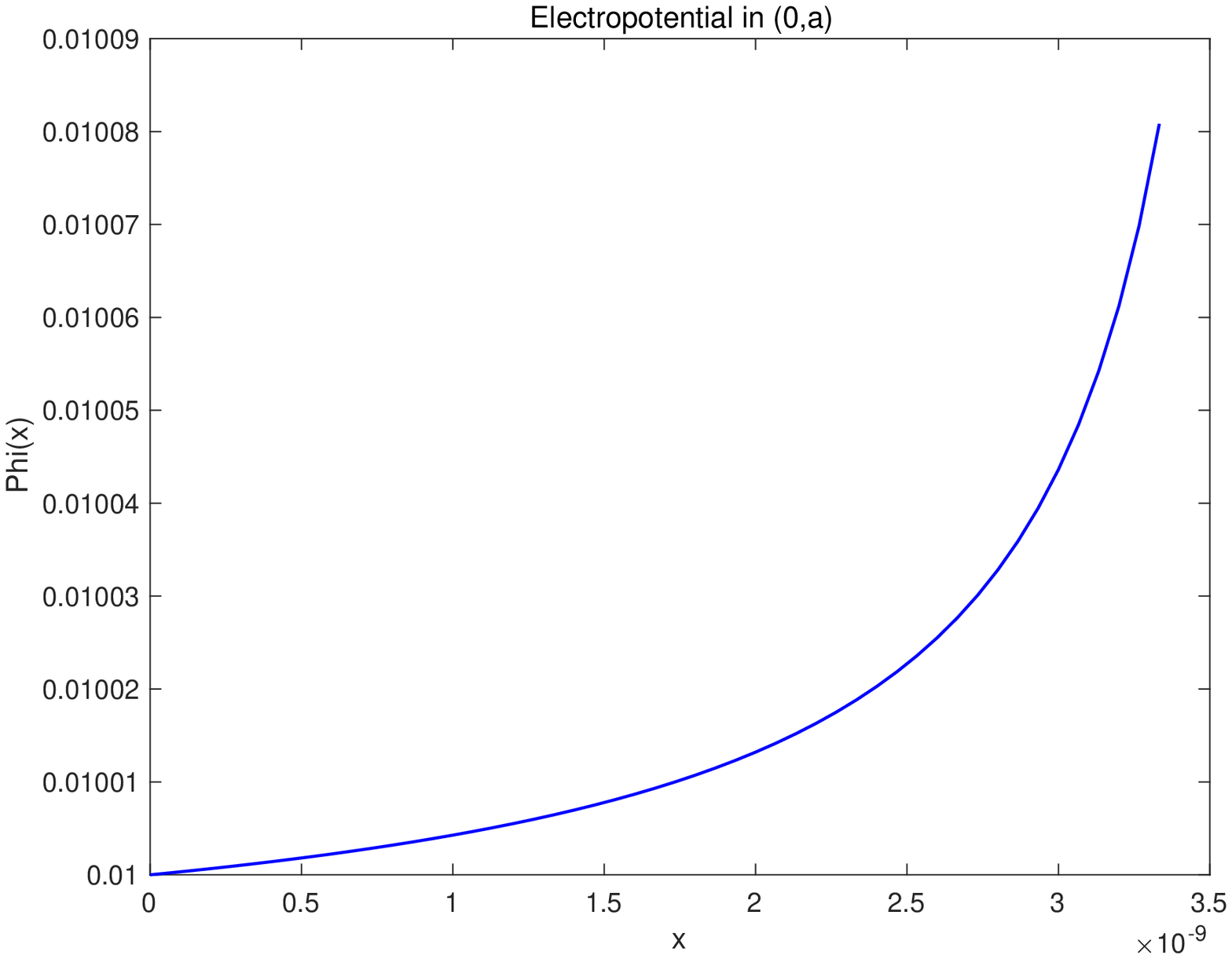}
\end{minipage}
\qquad
\begin{minipage}[t]{0.4\linewidth}
 \includegraphics[width=2.5in]{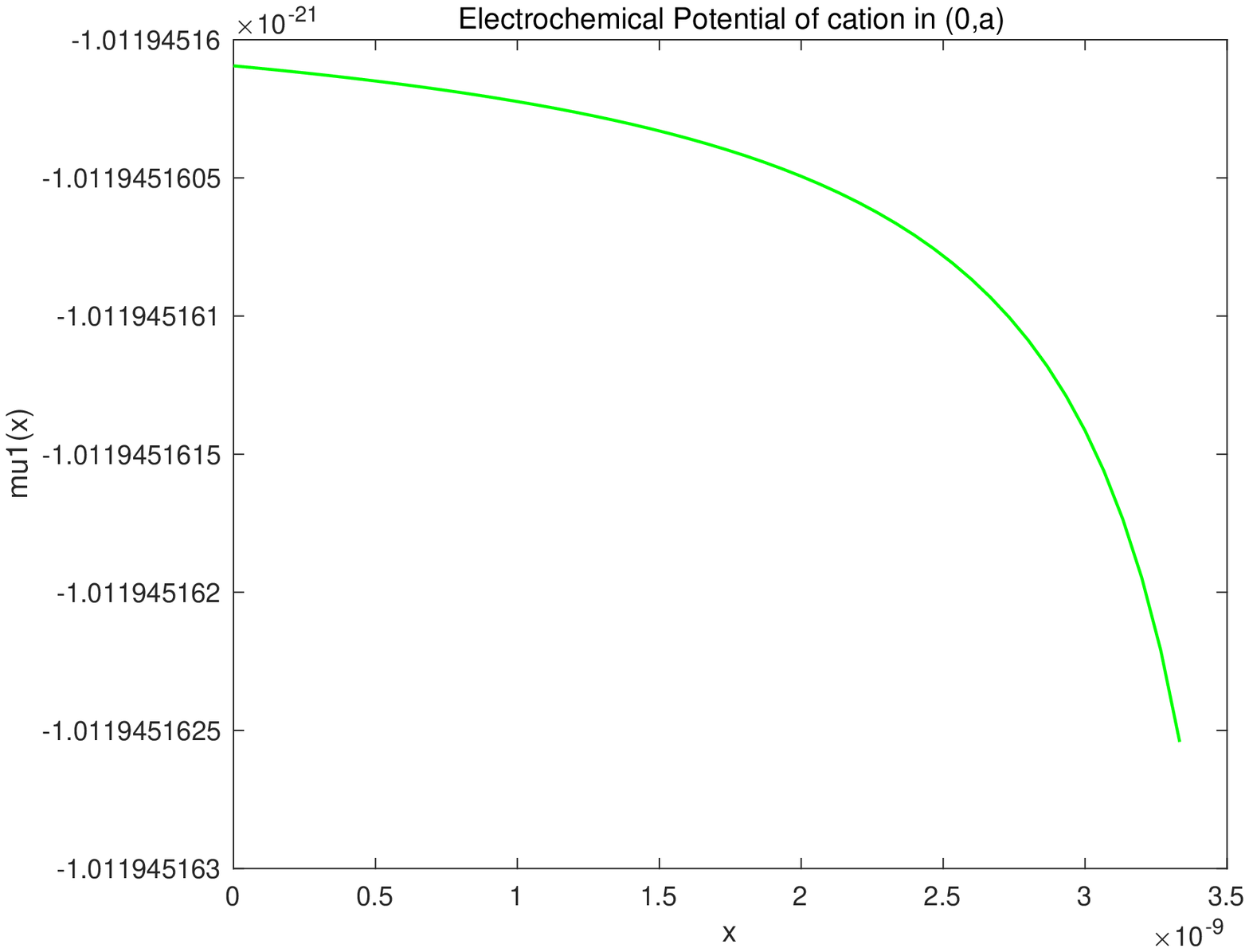}
 \end{minipage}
\caption{Profiles of   $c_1(x;\nu)$,  $\phi(x;\nu)$, and $\mu_1(x;\nu)$ over interval $[0,a]$. Note that $\mu_{10}'(x)=0$ but
$\mu_1'(x;\nu)= \mu_{11}'(x)\nu+O(\nu^2)\neq 0$ as shown in the figure.}
\end{figure}

  The following is then a direct consequence. 
\begin{cor}\label{claim11} Over the interval $(0,a)$,  
  $c_{10}(x)=O(1)$ and  $\mu_{10}'(x)=0$, and hence,  one has $J_{10}=0$.
 \end{cor}
 
  \subsection{Internal dynamics over the interval $(a,b)$}
  It follows again from \cite{ZL17} that one has the following approximations.
   \begin{prop}\label{expProfile21} For $x\in (a,b)$, 
 \begin{align*}
 \phi(x;\nu)=&-\ln \nu+\phi_0(x)+\phi_1(x)\nu +O(\nu^2)\;\mbox{ with }\\
 \phi_0(x)=&\ln{R}-2\ln{B_0}+\ln{2},\\
 \phi_1(x)=&\phi_1^a-A_0+\frac{J_{20}}{2}(H(x)-H(a));\\
  c_{10}(x)=&0,  \quad
 c_{11}(x)=\frac{1}{2}A_0^2-J_{11}(H(x)-H(a)),\\
 \mu_{10}(x)=&\ln{R}-2\ln{B_0}+\ln{2}+\ln{\left(\frac{1}{2}A_0^2-J_{11}(H(x)-H(a))\right)},\\
 \mu_{11}(x)=&\phi_1^a-A_0+\frac{J_{20}}{2}(H(x)-H(a)),
 \end{align*}
 where
 \begin{align*}
 A_0=&\frac{\sqrt{e^{V}L}}{(1-\beta)\sqrt{e^{V}L}+\alpha\sqrt{R}}((1-\beta)L+\alpha R),\\
 B_0=&\frac{\sqrt{R}}{(1-\beta)\sqrt{e^{V}L}+\alpha\sqrt{R}}((1-\beta)L+\alpha R),\\
 A_1=&\frac{2(\beta-\alpha)^2(L-A_0)^2-\alpha^2(A_0^2-B_0^2)\ln{\frac{B_0L}{A_0R}}}{4(\beta-\alpha)((1-\beta)L+\alpha R)(L-A_0)}A_0B_0,\\
 B_1=&-\frac{(1-\beta)\left(2(\beta-\alpha)^2(L-A_0)^2-\alpha^2(A_0^2-B_0^2)\ln{\frac{B_0L}{A_0R}}\right)}{4\alpha(\beta-\alpha)((1-\beta)L+\alpha R)(L-A_0)}A_0B_0,\\
 \phi_1^a=&\frac{\ln{\frac{B_0}{R}}}{\ln{\frac{B_0L}{A_0R}}}\left(2\Big(\frac{B_1}{B_0}-\frac{A_1}{A_0}\Big)+\frac{\beta-\alpha}{\alpha}(L-A_0)\right)-2\frac{B_1}{B_0}-\frac{\beta}{\alpha}(L-A_0)+L.
 \end{align*}
 \end{prop}
  Figure 2 includes the profiles of $c_1(x;\nu)$,  $\phi(x;\nu)$, and $\mu_1(x;\nu)$ over interval $[a,b]$.
  
 One has the following   immediate consequence.
 \begin{cor}\label{claim21} Over the interval $(a,b)$, 
   $c_{10}(x)=0$, and hence, $J_{10}=0$. Furthermore,   $\mu_{10}(a)=V+\ln L$ and $\mu_{10}(b)=\ln R$. As $R\to 0$, $\mu_{10}(a)-\mu_{10}(b)\to \infty$.   
\end{cor}
\begin{proof}  It follows from Proposition \ref{expProfile21} directly.  
 \end{proof}

  \begin{figure}[htbp]\label{c2one}
\centering
\begin{minipage}[t]{0.4\linewidth}
 \includegraphics[width=2.5in]{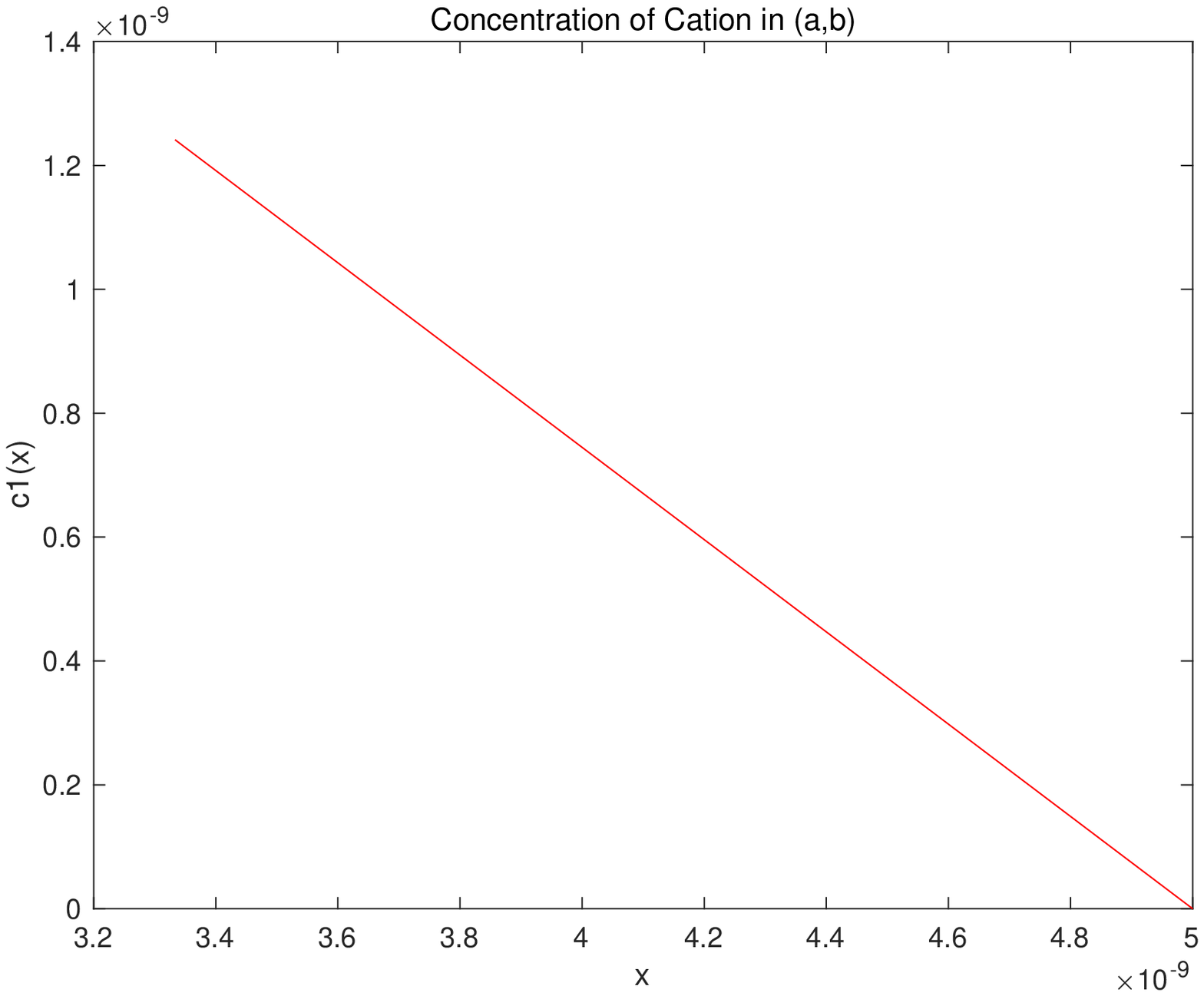}
 \end{minipage}
 \qquad
 \begin{minipage}[t]{0.4\linewidth}
 \includegraphics[width=2.5in]{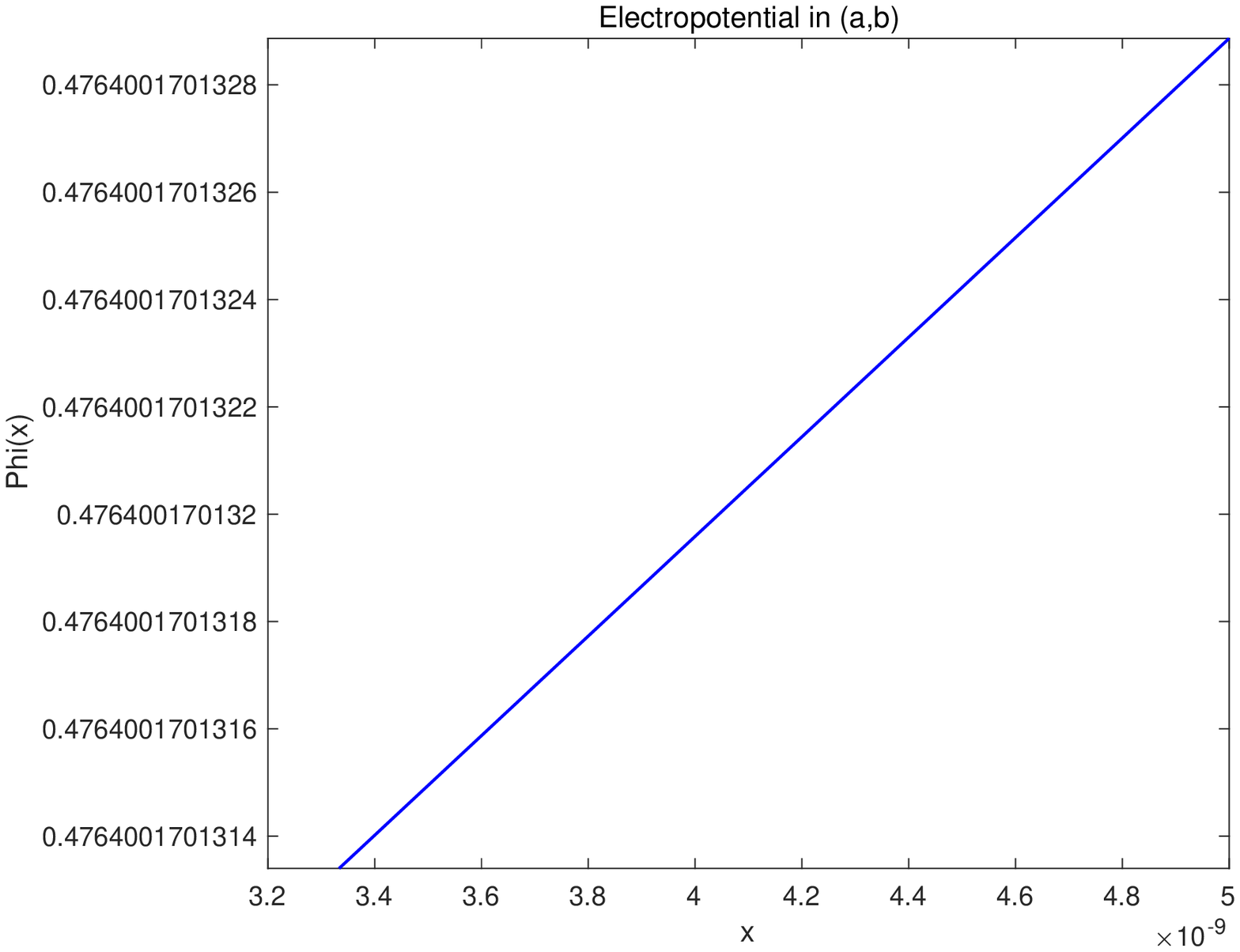}
\end{minipage}
\qquad
\begin{minipage}[t]{0.4\linewidth}
 \includegraphics[width=2.5in]{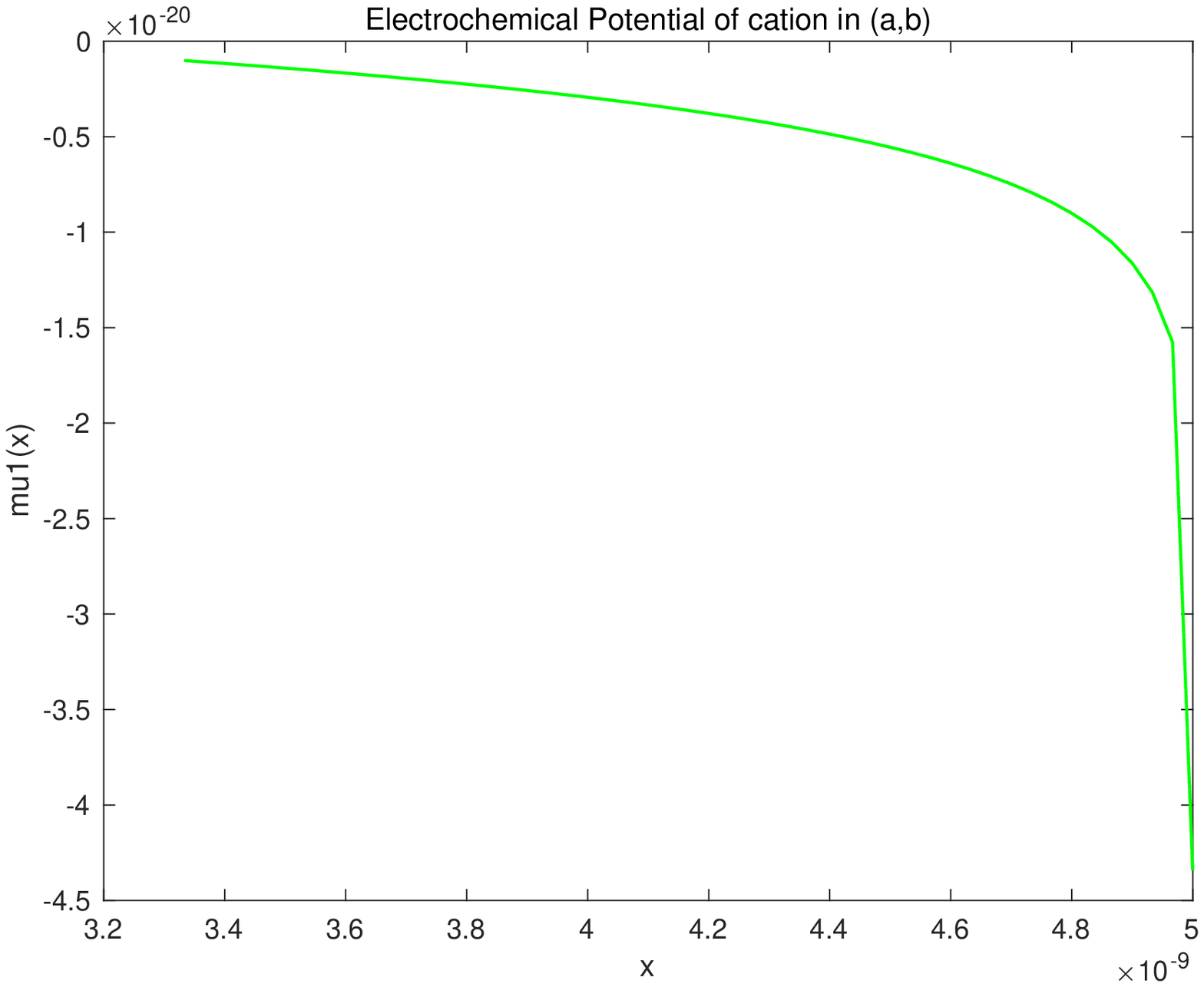}
 \end{minipage}
\caption{Profiles of  $c_1(x;\nu)$,  $\phi(x;\nu)$, $\mu_1(x;\nu)$ over interval $[a,b]$. Note that $c_{10}(x)=0$ but
$c_1(x;\nu)=c_{10}(x)+c_{11}(x)\nu+O(\nu^2)=c_{11}(x)\nu+O(\nu^2)\neq 0$ as shown in the figure. Also note the large drop of  
$\mu_1(x)$ over this interval due to that $R$ is taken to be small.}
\end{figure} 

  \subsection{Internal dynamics over the interval $(b,1)$}
   \begin{prop}\label{expProfile31} For $x\in (b,1)$,
 \begin{align*}
 \phi(x;\nu)=& \phi_0(x)+ \phi_1(x)\nu+O(\nu^2)\;\mbox{ with }\\
 \phi_0(x)=&\ln R-\ln c_{10}(x),\\
 \phi_1(x)=&\phi_1^b-\frac{1}{2}\left(B_0^2-\frac{2B_1-B_0^2}{B_0}\right)-\frac{B_0c_{11}(x)-B_1c_{10}(x)}{B_0c_{10}(x)}+\frac{2J_{11}}{J_{20}}\ln{\frac{c_{10}(x)}{B_0}};\\
 c_{10}(x)=& B_0-\frac{J_{20}}{2}(H(x)-H(b)),\quad
 c_{11}(x)= B_1-\frac{J_{11}+J_{21}}{2}(H(x)-H(b));\\
 \mu_{10}(x)=&\mu_{10}(1)=\ln{R},    \quad
 \mu_{11}(x)=\phi_1(x)+\frac{c_{11}(x)}{c_{10}(x)},
 \end{align*}
 where
 \begin{align*}
\phi_1^{b}=&\frac{\ln{\frac{B_0}{R}}}{\ln{\frac{B_0L}{A_0R}}}\left(2\Big(\frac{B_1}{B_0}-\frac{A_1}{A_0}\Big)+\frac{\beta-\alpha}{\alpha}(L-A_0)\right)-2\frac{B_1}{B_0}+B_0.
 \end{align*}
 \end{prop}
      
   Figure 3 shows the profiles of $c_1(x;\nu)$,  $\phi(x;\nu)$, and $\mu_1(x;\nu)$ over interval $[b,1]$.
   \begin{figure}[htbp]\label{c1three}
\centering
\begin{minipage}[t]{0.4\linewidth}
 \includegraphics[width=2.5in]{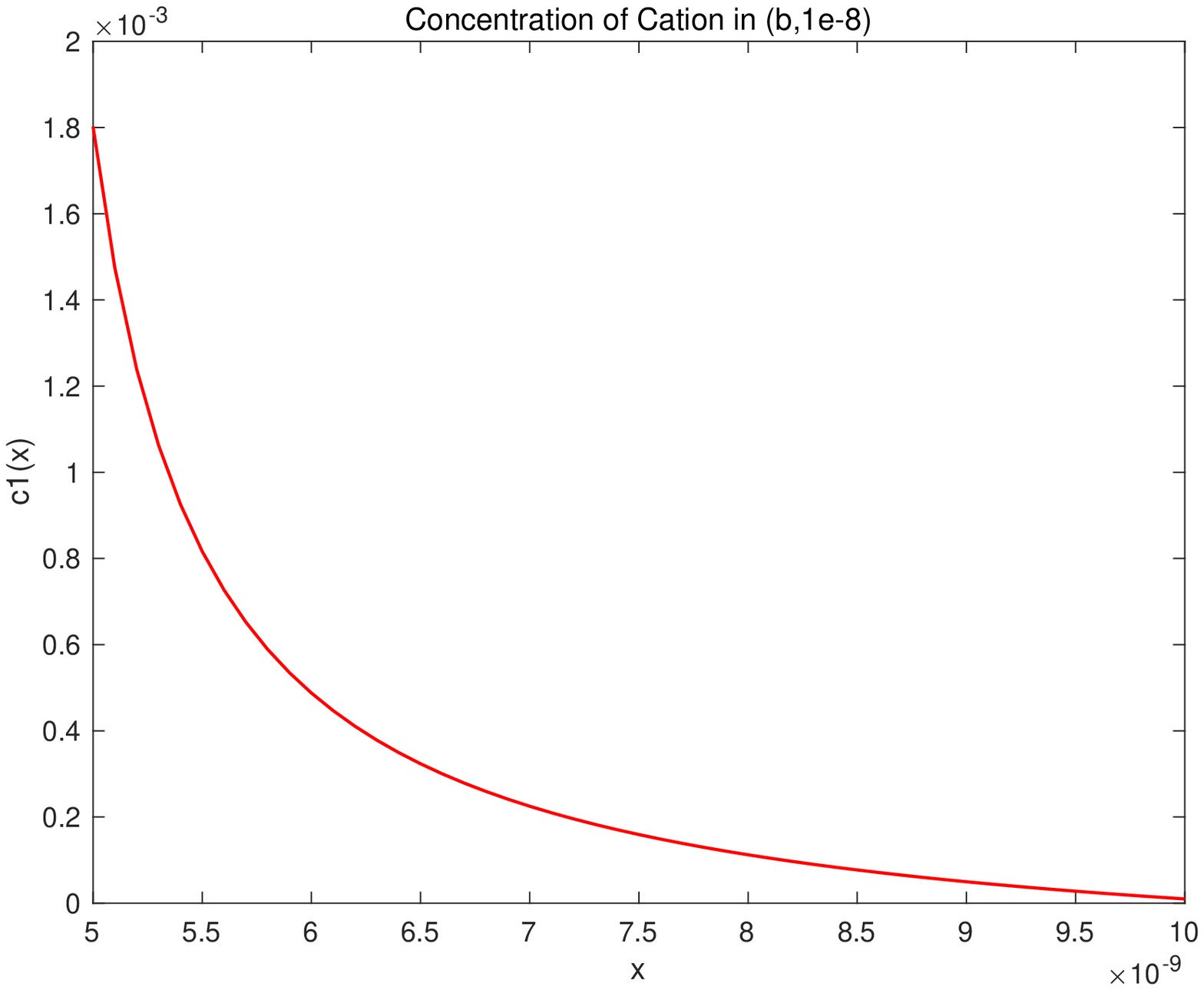}
 \end{minipage}
\qquad
 \begin{minipage}[t]{0.4\linewidth}
 \includegraphics[width=2.5in]{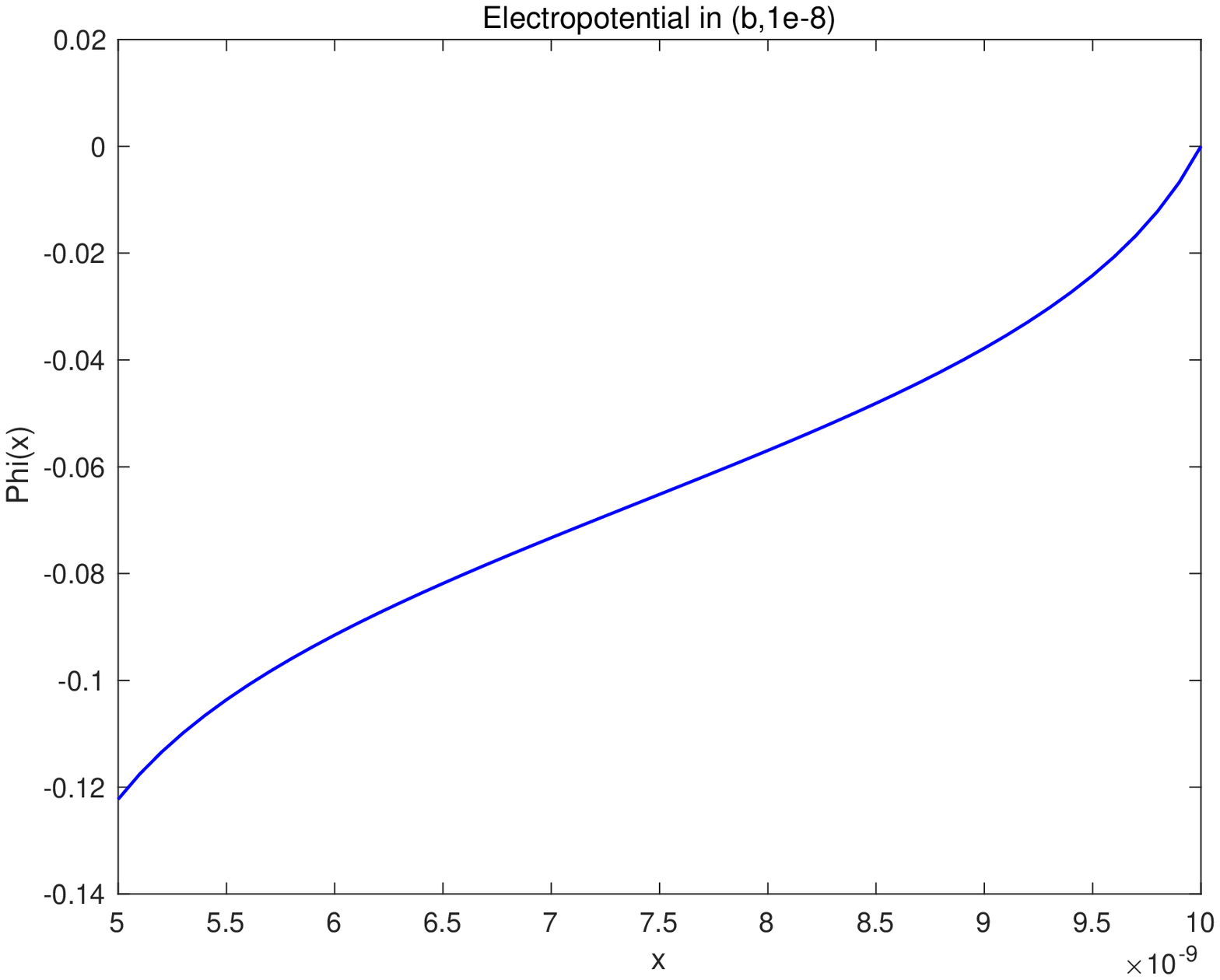}
\end{minipage}
\qquad
\begin{minipage}[t]{0.4\linewidth}
 \includegraphics[width=2.5in]{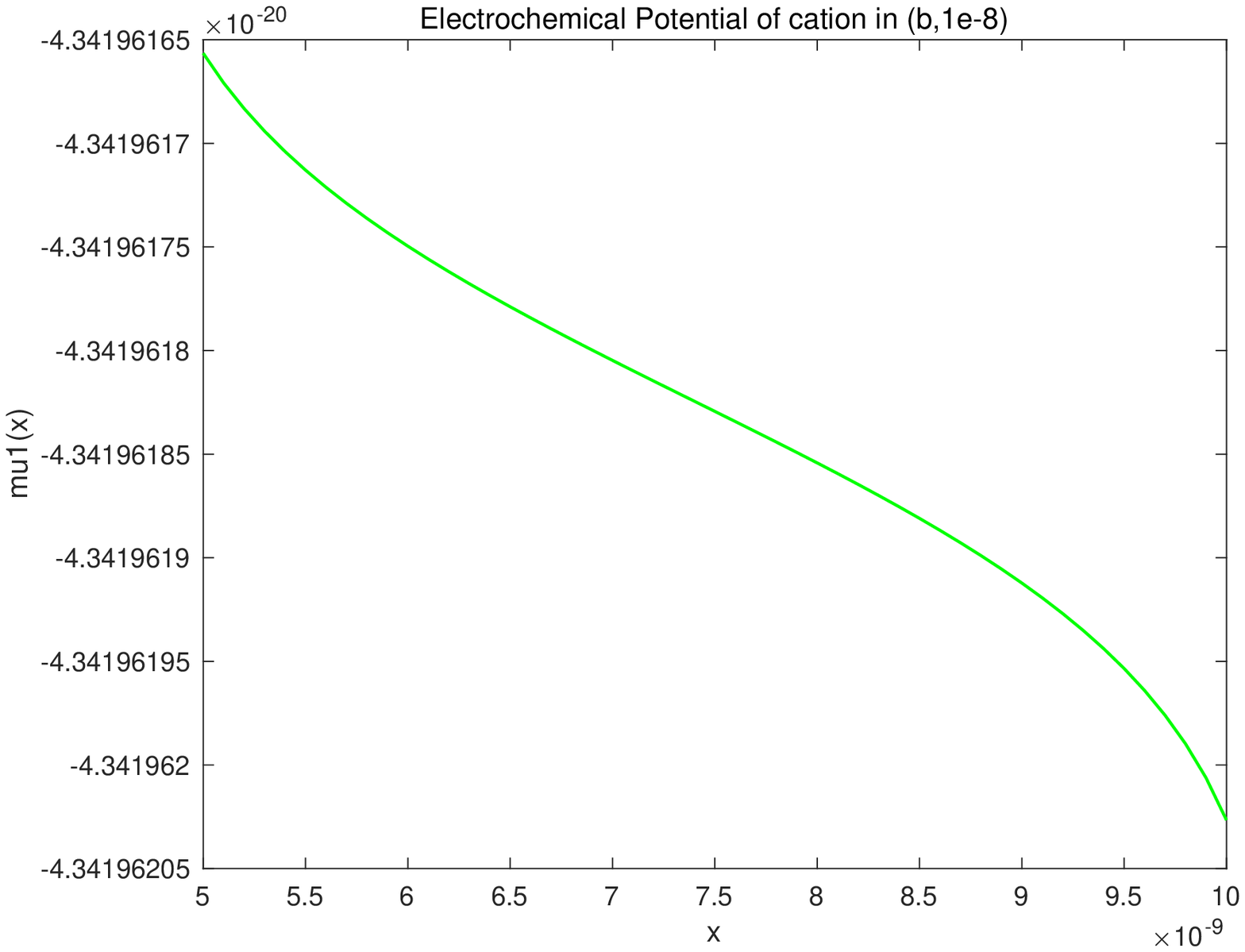}
 \end{minipage}
\caption{Profiles of  $c_1(x;\nu)$,  $\phi(x;\nu)$, and $\mu_1(x;\nu)$ on interval $[b,1]$. Note that $\mu_{10}'(x)=0$ but $\mu_1'(x;\nu)=\mu_{10}(x)+ \mu_{11}'(x)\nu+O(\nu^2)= \mu_{11}'(x)\nu+O(\nu^2)\neq 0$.}
\end{figure}

 \begin{cor}\label{claim31} Over the interval $(b,1)$,   noting $B_0=O(\sqrt{R})$,
  $c_{10}(x)=O(1)$ and $\mu_1'(x)=0$, and hence, $J_{10}=0$.
    \end{cor}

     \subsection{Summary of mechanism for $J_{10}=0$.}
From the above discussion, we conclude that the mechanisms for $J_{10}=0$ are different over each subintervals. More precisely, 
one has that, 
\begin{itemize}
\item over the first interval $(0,a)$, $J_{10}=0$ is the result of {\em constant} electrochemical potential $\mu_{10}(x)=\mu_{10}(0)=\mu_{10}(a)$ (so that $\mu_{10}'(x)=0$)    while $c_{10}(x)\neq 0$; 
\item over the last interval $(b,1)$, similar to that over the first interval $(0,a)$,  $\mu_{10}(x)=\mu_{10}(1)$ so that $\mu_{10}'(x)=0$ while $c_{10}(x)\neq 0$, and, for $R$ small, $c_{10}(x)=O(\sqrt{R})$ over this interval; 
\item over the middle interval $(a,b)$ where permanent charge is not zero and large, the electrochemical potential $\mu_{10}(x)$ is not a constant so that $\mu_{10}'(x)\neq 0$, however, $c_{10}(x)=0$, in particular, the drop of $\mu_{10}(x)$ over this interval equals the transmembrane electrochemical potentials, that is,  $\mu_{10}(a)-\mu_{10}(b)=\mu_1(0)-\mu_1(1)$.
\end{itemize} 

Here we provide the profiles   of concentration (Fig. 4), electrical potential (Fig. 5) and electrochemical potential (Fig. 6) over the whole interval $[0,1]$.

The figures of $c_1(x)$ and $\phi(x)$ over interval $(0,x_0)$ are not continue, because we make the plots of system \eqref{ssPNP} with $\varepsilon=0$. For the limiting system at $x=a$ and $x=b$, there are two fast layers, $c_1(x)$ and $\phi(x)$ changes very fast, but $\mu_1(x)$ keeps the same value in fast layers.
  Recall  that $\mu_1^\delta=\mu_1(0)-\mu_1(1)=k_B T(V+\ln{L}-\ln{R})$, if $R=10^{-5}$, $\mu_1^\delta \sim 2.5902\times 10^{-23}$.

\begin{figure}[htbp]\label{c1}
\centering
  \begin{minipage}[t]{0.6\linewidth}
 \includegraphics[width=3in]{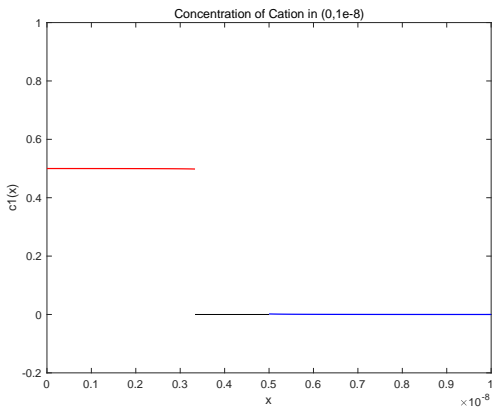}
\end{minipage}
\caption{Profile  of $c_1(x;\nu)$  over $[0,1]$. Note that $c_{10}(x)=0$ for $x\in (a,b)$  and, for $x\in (b,1)$, $c_{10}(x)=O(\sqrt{R})$  is not zero but small for   $R=10^{-5}$ chosen for the numerics.}
\end{figure}

\begin{figure}[htbp]\label{phi}
\centering
 \begin{minipage}[t]{0.6\linewidth}
 \includegraphics[width=3in]{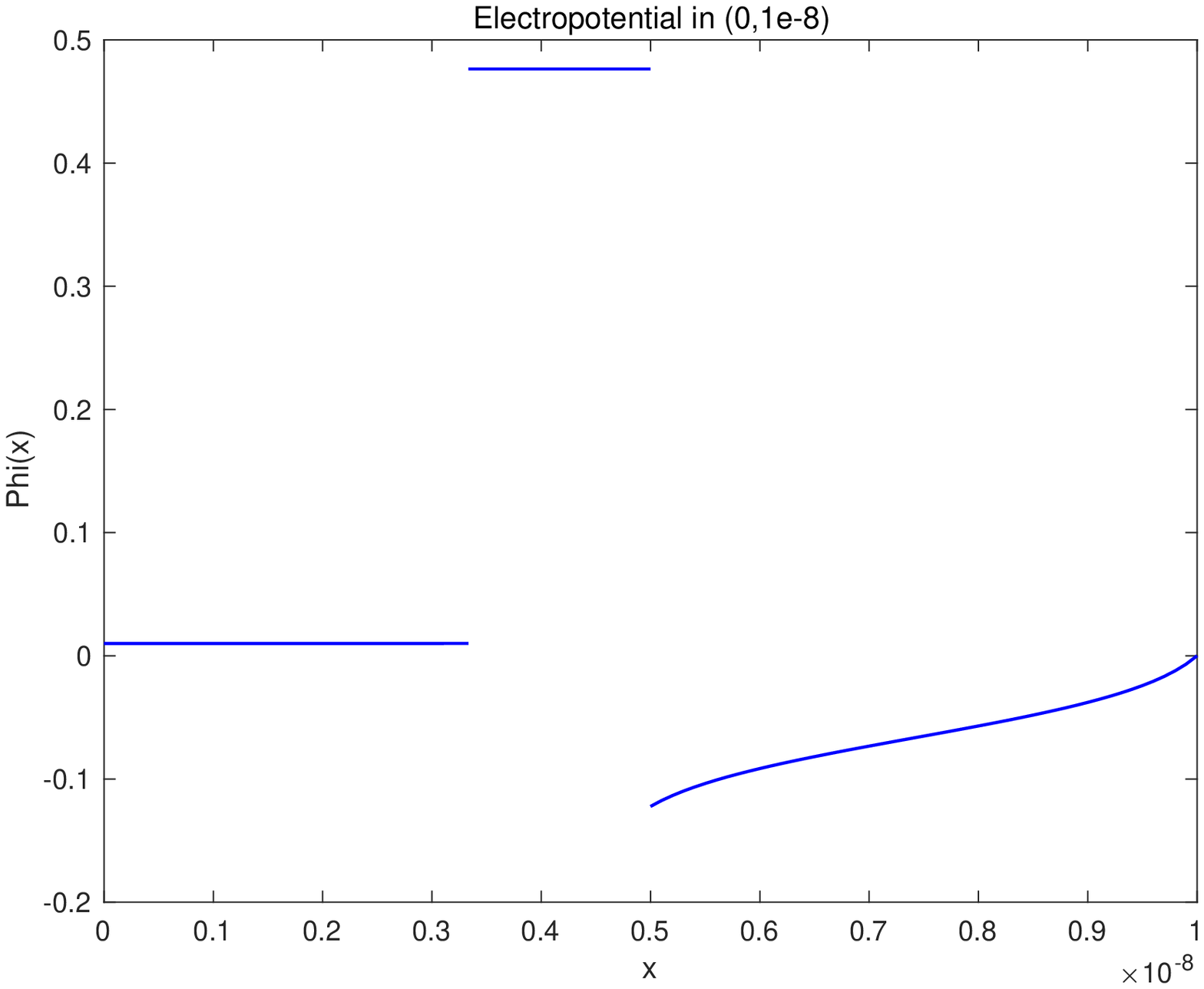}
\end{minipage}
\caption{Profile of $\phi(x;\nu)$ over $[0,1]$ }
\end{figure}

\begin{figure}[htbp]\label{mu1}
\centering
  \begin{minipage}[t]{0.6\linewidth}
 \includegraphics[width=3in]{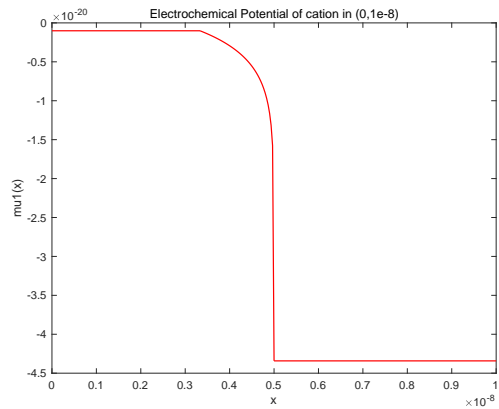}
\end{minipage}
\caption{Profile of $\mu_1(x;\nu)$ over   $[0,1]$. Note the large drop of $\mu_1$ over the interval $(a,b)=(1/3,1/2)$ and the nearly constant values over the other two subintervals.}
\end{figure}

\newpage
\section{Declining phenomenon and internal dynamics}\label{Dyn4declin}
\setcounter{equation}{0}
In this section, we will show that {\em large permanent charge is a mechanism for the   declining phenomenon}   described in the introduction. Recall that, by {\em the declining phenomenon}, we mean the following observed experimentally. 
 
{\em For fixed $V$ and $L_1=L_2=L$, as $R_1=R_2=R$ decreases to zero, the flux of counterion ($J_2$ in the setting   since $Q_0>0$) decreases monotonically to zero.}
\medskip

 This phenomenon is rather  counterintuitive since the transmembrane  electrochemical potential for the counter-ion 
\[ \mu_2(0)-\mu_2(1)=z_2V+\ln L_2-\ln R_2\to +\infty\;\mbox{ as }\; R_2=R \to 0.\] 

\subsection{Experimental phenomena are consistent with our analysis}

  The result in (\ref{ejJ}) actually justifies the observation, up to the leading order $J_{20}$ in $\nu$ for $\nu$ near $0$ (or for $Q_0\to +\infty$); that is,
  \begin{prop}\label{declining} The leading order term $J_{20}$ of the flux is monotone and concave downward as a function of $R$ and, as $R\to 0$, $J_{20}\to 0$. 
  \end{prop}
\begin{proof}   Indeed, from the expression of $J_{20}$ in (\ref{ejJ}) and treating $J_{20}$ as a function of $w$ where $R=w^2$ for convenience, one has
\[J_{20}(w)=\frac{\sqrt{L}}{2H(1)}\frac{\sqrt{e^{-V}L}w-w^2}{(1-\beta)\sqrt{L}+\alpha\sqrt{e^{-V}}w}.\]
It is clear that $J_{20}(w)\to 0^+$ as $w\to 0^+$. Note that the derivative of $J_{20}$ in $w$ is 
\[ J_{20}'(w)=\frac{1}{2H(1)}\frac{(1-\beta)L\sqrt{e^{-V}}-2(1-\beta)\sqrt{L}w-\alpha \sqrt{e^{-V}} w^2}{[(1-\beta)\sqrt{L}+\alpha\sqrt{e^{-V}}w]^2}.\]
 Thus, from the expression of the numerator, 
 \medskip
 
 \centerline{\Large \em if $w$ is smaller than some $w_0$, then $J_{20}'(w)>0$,}
 \medskip
 
 \noindent
 and hence, as $w\to 0^+$ (or equivalently, $R_2=R\to 0$) monotonically over the interval $[0,w_0]$, $J_{20}(w)\to 0$ monotonically.  
 
It is not hard to show that for $R_2=R\ge 0$ but smaller than some positive value, the graph of $J_{20}$ as a function of $R$ is  concave downward.
\end{proof}
In Figure 7,   the horizontal axis is for $R$ and the vertical for $J_{20}$. We fix $V=0.425$ and $L=0.5$ as in (\ref{parameters}),  and vary $R\in (0,10^{-3}]$. 
   The monotonicity and concave downward features of the graph are apparent. 
\medskip

\begin{figure}[htbp]
\centering
\begin{minipage}[t]{0.6\linewidth} 
 \includegraphics[width=3in]{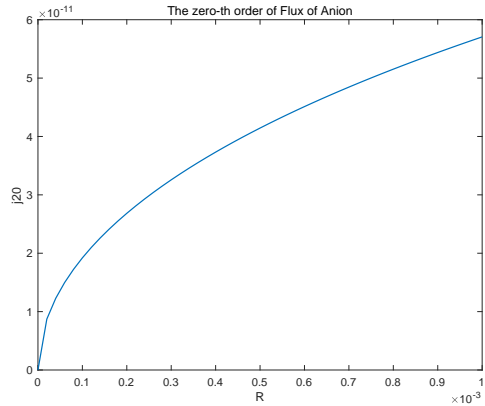}
 \end{minipage}
  \caption{Declining curve:  $J_{20}$ vs  $R$ for $R\in (0, 10^{-3}]$ with $L=0.5$ and $V=0.425$.}
\end{figure}

\begin{rem}\label{largemu}  We comment that, for large permanent charge, the declining curve phenomenon occurs when the transmembrane electrochemical potential $\mu_2(0)-\mu_2(1)$ is increasing to infinity in a particular way; that is, as $R\to 0$. If one increases  the transmembrane electrochemical potential $\mu_2(0)-\mu_2(1)$ in a different manner, for example, as $|V|\to \infty$ or as $L\to \infty$,  Corollary \ref{Saturate} shows that the declining curve phenomenon does not happen.

 It is also important to note that, when   the next order term $J_{21}\nu$ is considered, then, as $R_2=R\to 0$, $J_{21}\to \infty$.  But, if $R\to 0$ and $\nu\ln R\to 0$, then $J_{21}\nu\to0$. Thus, only when $\nu$ is very small ($Q_0$ is very large), is the term $J_{21}\nu$ not significant, and hence, the term $J_{20}$ dominates the described behavior. \qed
\end{rem}
\bigskip

 \subsection{Mechanism of declining phenomena from the profiles}
Recall, from the Nernst-Planck equation in (\ref{PNP}) that 
\[-J_2=\frac{1}{k_BT}D_2(x)h(x) c_2(x)\frac{d}{dx}\mu_2(x).\]
 Since $D_2(x)$ and $h(x)$ are fixed, we will treat them as of order $O(1)$ quantities so that they do not contribute much to the near zero flux scenario as $R\to 0$.
Thus, as far as the near zero flux mechanism is concerned, one has
\begin{align}\label{approNP}
-J_2\approx c_2(x)\frac{d}{dx}\mu_2(x).
\end{align}

One sees that the gradient of the electrochemical potential $\frac{d}{dx}\mu_2(x)$ is the main driving force for the flux $J_2$. Intuitively, large drop of (or transmembrane) electrochemical potential $\mu_2(0)-\mu_2(1)$ of $\mu_2$ produces large flux $J_2$. In this sense, the declining curve phenomenon is rather counterintuitive. A careful look  at (\ref{approNP}) reveals that there is only one possibility for the declining curve phenomenon; that is, whenever  $\frac{d}{dx}\mu_2(x)$ is  large,  $c_2(x)$ has to be much smaller in order to produce a small flux $|J_2|$. We will apply the analytical results of  the internal dynamics from (\cite{ZL17}) to show that this is indeed the case.   

For our setup,  there are three regions of permanent charge $Q(x)$: $Q(x)=0$ for $x\in [0,a)$ and $x\in (b,1]$ and $Q(x)=2Q_0$ for $x\in [a,b]$ with {\em large} $Q_0$ or {\em small} $\nu$
with ${\nu}=1/{Q_0}$.
 
 For fixed $V$ and $L$, $\mu_2(0)-\mu_2(1)=z_2V+\ln L-\ln R\approx -\ln R\gg 1$ for small $R$.
  We need to understand   
 
 (i) HOW the electrochemical potential $\mu_2$ drops an order $O(-\ln R)\gg 1$ over the interval $x\in [0,1]$;
 
  (ii) HOW
 $J_2$ can be small, for small   $\nu$ and $R$,   at every $x\in [0.1]$;
 
 (iii) Most importantly, HOW the above two things, with the constraint  (\ref{approNP}),  can happen simultaneously.
 
 There are two small parameters $\nu$ and $R$ in our considerations. The relative sizes of these two parameters is relevant for the result. We will assume $\nu\ln R\ll 1$. 
 
We now discuss what happens over each subinterval     based on the approximated (of zeroth order in $\varepsilon$) functions of profiles. To do so, let   
\[(\phi(x;\varepsilon,\nu),c_k(x;\varepsilon,\nu), J_k(\varepsilon,\nu))=(\phi(x;\nu),c_k(x;\nu) , J_k(\nu))+O(\varepsilon)\] 
be the   solution of the boundary value problem. For   $\nu>0$ small, one has the following expansions
\begin{align*}
J_2(\nu)=J_{20}+J_{21}\nu+O(\nu^2),
 \end{align*}
 where   $J_{20}$ and $J_{21}$ are given in (\ref{ejJ}). The expansions in $\nu$ for 
 $\phi(x;\nu)$,  $c_1(x;\nu)$, and  $c_2(x;\nu)$ are not {\em regular} and are qualitatively different  over the subintervals $(0,a)$, $(a,b)$ and $(b,1)$. They will be   given explicitly in each subsection below for us to understand what happens over each subinterval. 
 
     \subsubsection{Internal dynamics over the interval $(0,a)$}
   The leading order terms of $(\phi,c_2)$  are derived in \cite{ZL17}. One has
      \begin{prop}\label{expProfile12} For $x\in (0,a)$, 
  \begin{align*}
  \phi(x;\nu)=&\phi_0(x)+\phi_1(x)\nu+O(\nu^2) \mbox{ with }\\
  \phi_0(x)=& V- \ln\Big(1- \frac{J_{20}}{2L}H(x)\Big),\quad
 \phi_1(x)=-\frac{c_{21}(x)}{c_{20}(x)}+\frac{2J_{11}}{J_{20}}\ln{\frac{c_{20}(x)}{L}};\\
c_{20}(x)=&  L- \frac{J_{20}}{2}H(x),\quad
 c_{21}(x)= -\frac{J_{11}+J_{21}}{2}H(x);\\
 \mu_{20}(x)=&-V+2\ln{c_{20}(x)}-\ln{L},\quad
 \mu_{21}(x)=2\left(\frac{c_{21}(x)}{c_{20}(x)}-\frac{J_{11}}{J_{20}}\ln{\frac{c_{20}(x)}{L}}\right).
 \end{align*}
 \end{prop}
  
  Figure 8 shows profiles of $c_2(x;\nu)$, $\phi(x;\nu)$ and $\mu_2(x;\nu)$ over the interval $(0,a)$.  
 \begin{cor}\label{claim12} Over the interval $(0,a)$,  
 $c_2(x;\nu)=O(1)$  BUT 
 \[\mu_2'(x;\nu)=-\frac{J_{20}}{h(x)c_{20}(x)}+O(\nu\ln R)=O(\sqrt{R},\nu\ln R).\]
  Therefore,  from (\ref{approNP}), 
     \[J_2\approx c_2(x) \frac{d}{dx}\mu_2(x;\nu)=O(1)O(\sqrt{R}, \nu\ln R)=O(\sqrt{R},\nu\ln R)\]
     and $\mu_2(x)$   drops an order of  $O(\sqrt{R},\nu\ln R)$ over the interval $(0,a)$.
    
 In particular, for zeroth order in $\nu$, $J_{2}\approx J_{20}=O(\sqrt{R})$ and $\mu_2'(x)=O(\sqrt{R})$  over the interval $(0,a)$.
  \end{cor}
  
  \begin{figure}[htbp]\label{c2one}
\centering
  \begin{minipage}[t]{0.4\linewidth}
 \includegraphics[width=2.5in]{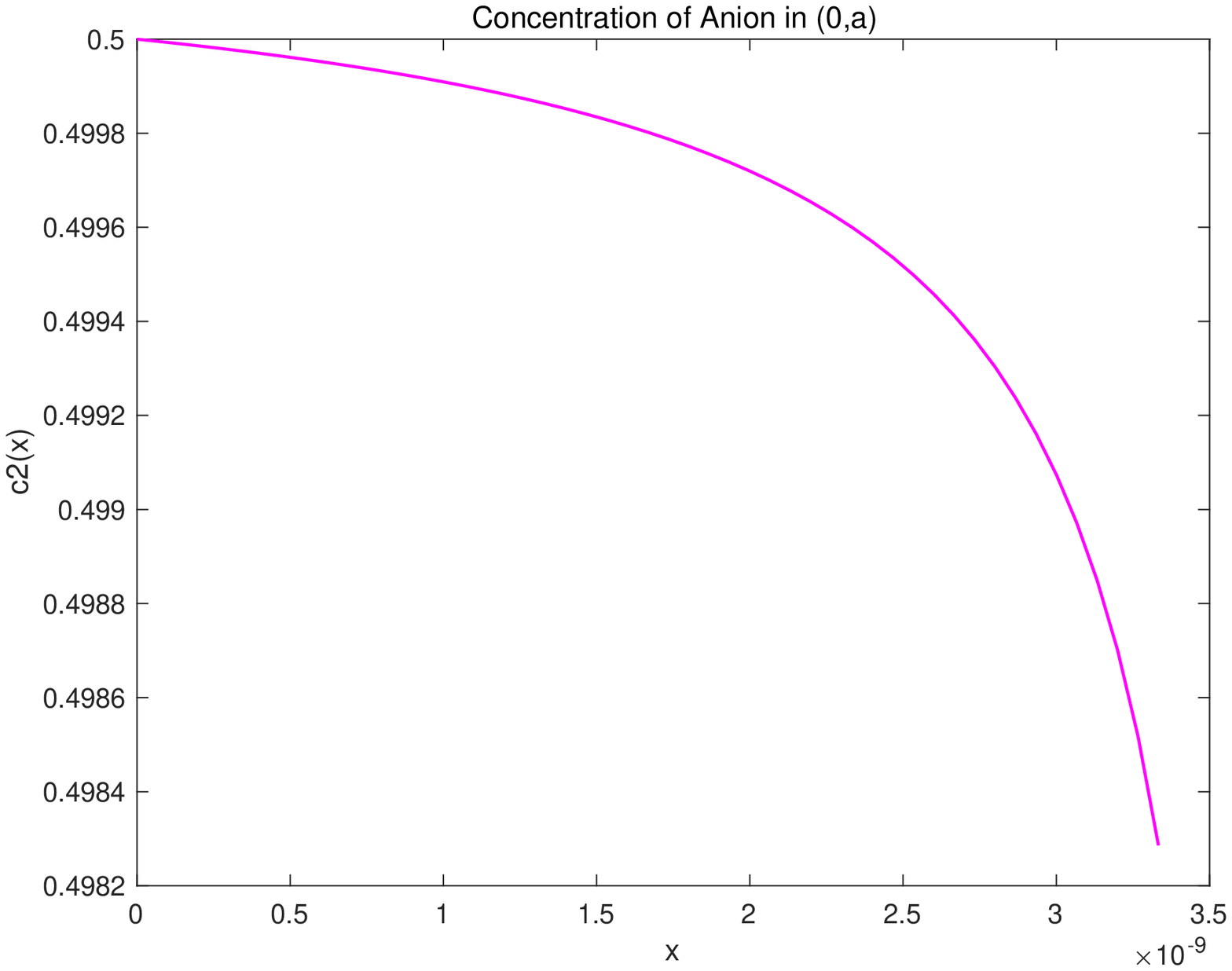}
\end{minipage}
\qquad
 \begin{minipage}[t]{0.4\linewidth}
 \includegraphics[width=2.5in]{phi_0a.eps}
\end{minipage}
 \qquad
 \begin{minipage}[t]{0.4\linewidth}
 \includegraphics[width=2.5in]{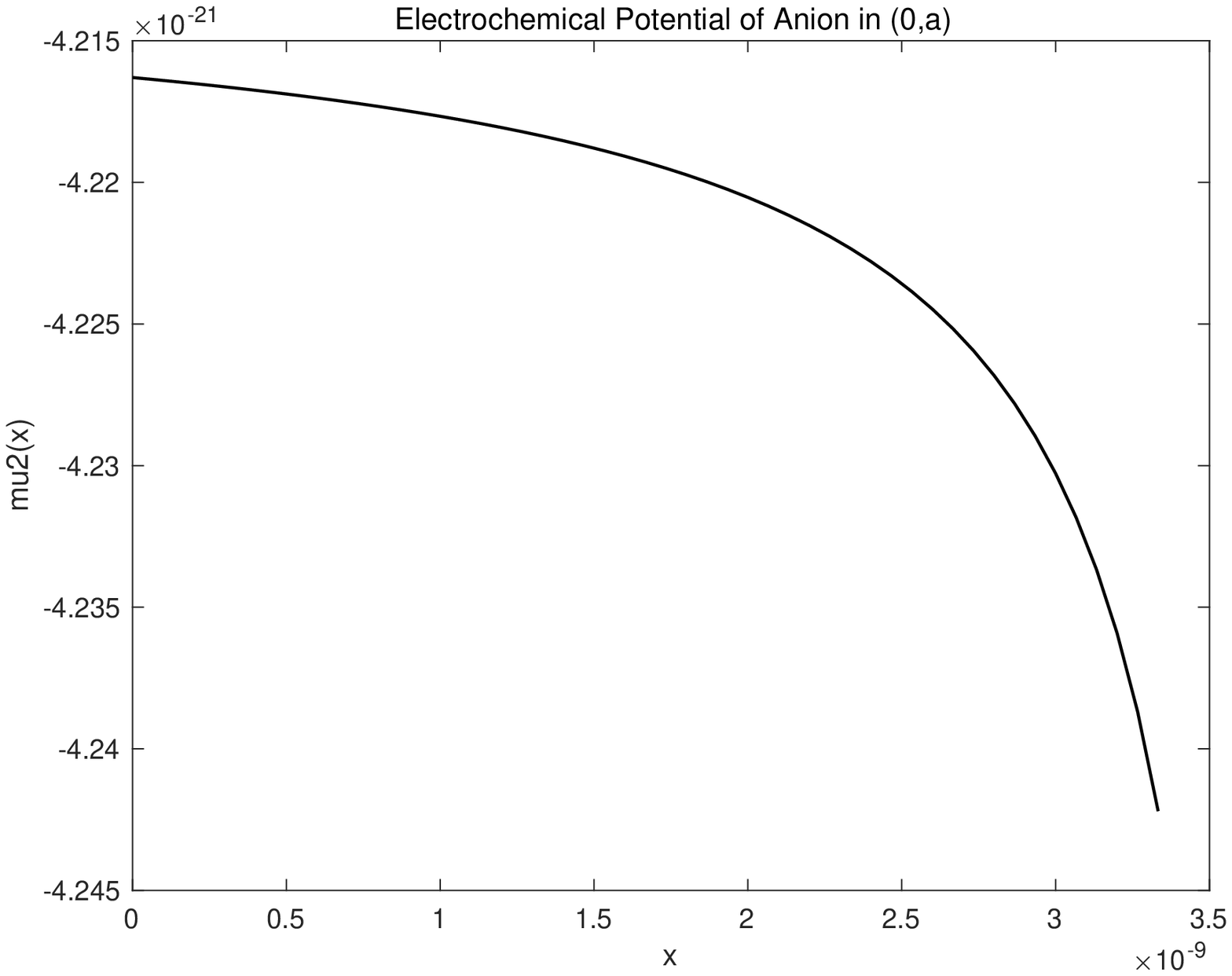}
\end{minipage}
\caption{Profiles of  $c_k(x;\nu)$,  $\phi(x;\nu)$, and  $\mu_k(x;\nu)$ over interval $[0,a]$.}
\end{figure} 

\newpage

  \subsubsection{Internal dynamics over the interval $(a,b)$}
  It follows from \cite{ZL17} that 
   \begin{prop}\label{expProfile22} For $x\in (a,b)$, 
 \begin{align*}
 \phi(x;\nu)=&-\ln \nu+\phi_0(x)+\phi_1(x)\nu +O(\nu^2)\;\mbox{ with }\\
 \phi_0(x)=&\ln{R}-2\ln{B_0}+\ln{2},\\
 \phi_1(x)=&\phi_1^a-A_0+\frac{J_{20}}{2}(H(x)-H(a));\\
 c_{2}(x;\nu)=&\frac{1}{\nu}+\left(\frac{1}{2}A_0^2-J_{11}(H(x)-H(a))\right)\nu+O(\nu^2);\\
 \mu_{20}(x)=&-\ln{R}+2\ln{B_0}-\ln{2},\\
 \mu_{21}(x)=&-\phi_1^a+A_0-\frac{J_{20}}{2}(H(x)-H(a)),\\
 \mu_{22}(x)=&\frac{1}{2}A_0^2-J_{11}(H(x)-H(a));
 \end{align*}
 where
 \begin{align*}
 A_0=&\frac{\sqrt{e^{V}L}((1-\beta)L+\alpha R)}{(1-\beta)\sqrt{e^{V}L}+\alpha\sqrt{R}},\quad
 B_0=\frac{\sqrt{R}((1-\beta)L+\alpha R)}{(1-\beta)\sqrt{e^{V}L}+\alpha\sqrt{R}},\\
 A_1=&\frac{2(\beta-\alpha)^2(L-A_0)^2-\alpha^2(A_0^2-B_0^2)\ln{\frac{B_0L}{A_0R}}}{4(\beta-\alpha)((1-\beta)L+\alpha R)(L-A_0)}A_0B_0,\\
 B_1=&-\frac{(1-\beta)\left(2(\beta-\alpha)^2(L-A_0)^2-\alpha^2(A_0^2-B_0^2)\ln{\frac{B_0L}{A_0R}}\right)}{4\alpha(\beta-\alpha)((1-\beta)L+\alpha R)(L-A_0)}A_0B_0,\\
 \phi_1^a=&\frac{\ln{\frac{B_0}{R}}}{\ln{\frac{B_0L}{A_0R}}}\left(2\Big(\frac{B_1}{B_0}-\frac{A_1}{A_0}\Big)+\frac{\beta-\alpha}{\alpha}(L-A_0)\right)-2\frac{B_1}{B_0}-\frac{\beta}{\alpha}(L-A_0)+L.
 \end{align*}
 \end{prop}
 
 The profiles of $c_2(x;\nu)$,  $\phi(x;\nu)$, and $\mu_2(x;\nu)$ are shown in Figure 9.
  \begin{figure}[htbp] 
\centering
  \begin{minipage}[t]{0.4\linewidth}
 \includegraphics[width=2.5in]{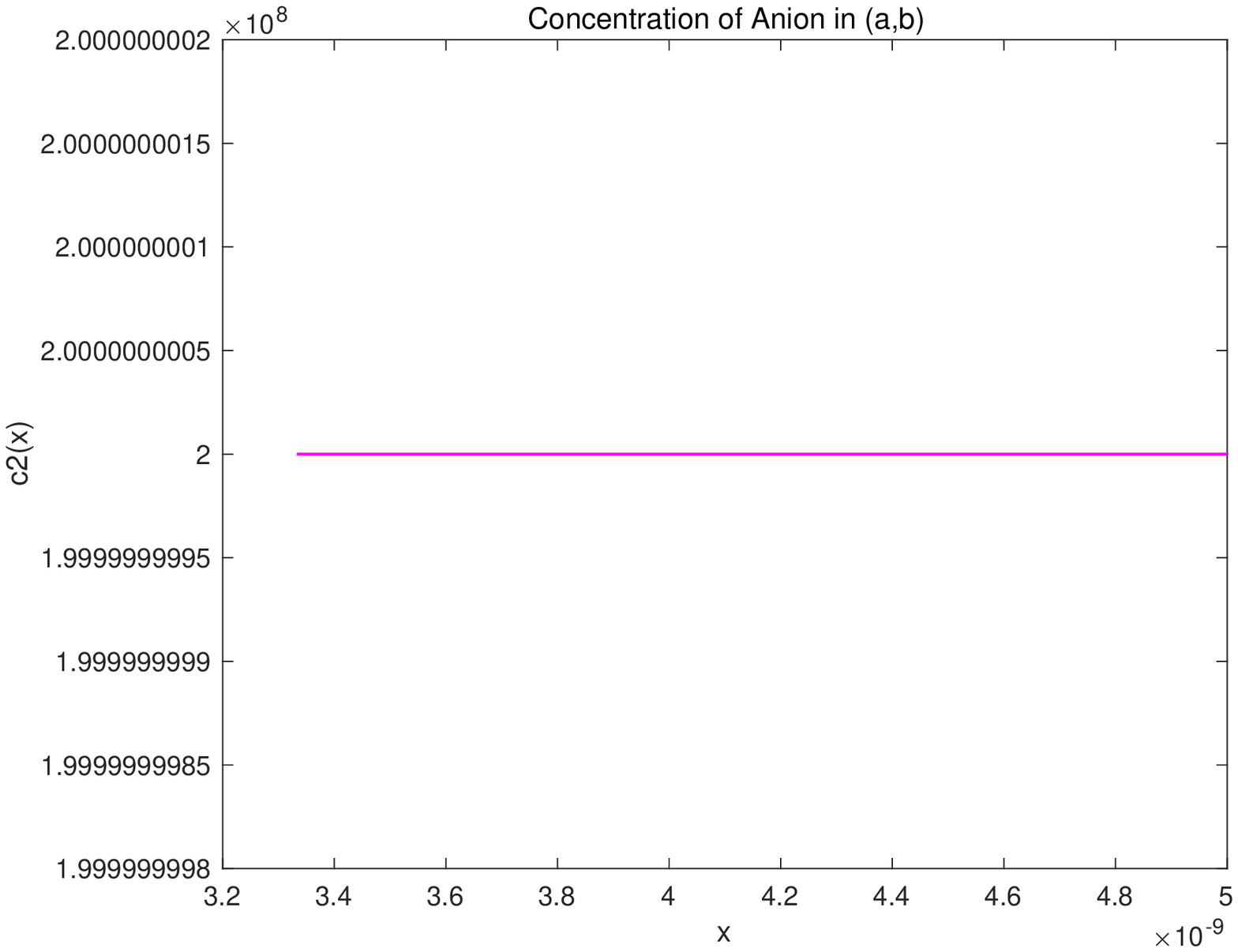}
\end{minipage}
\qquad
 \begin{minipage}[t]{0.4\linewidth}
 \includegraphics[width=2.5in]{phi_ab.eps}
\end{minipage}
\qquad
 \begin{minipage}[t]{0.4\linewidth}
 \includegraphics[width=2.5in]{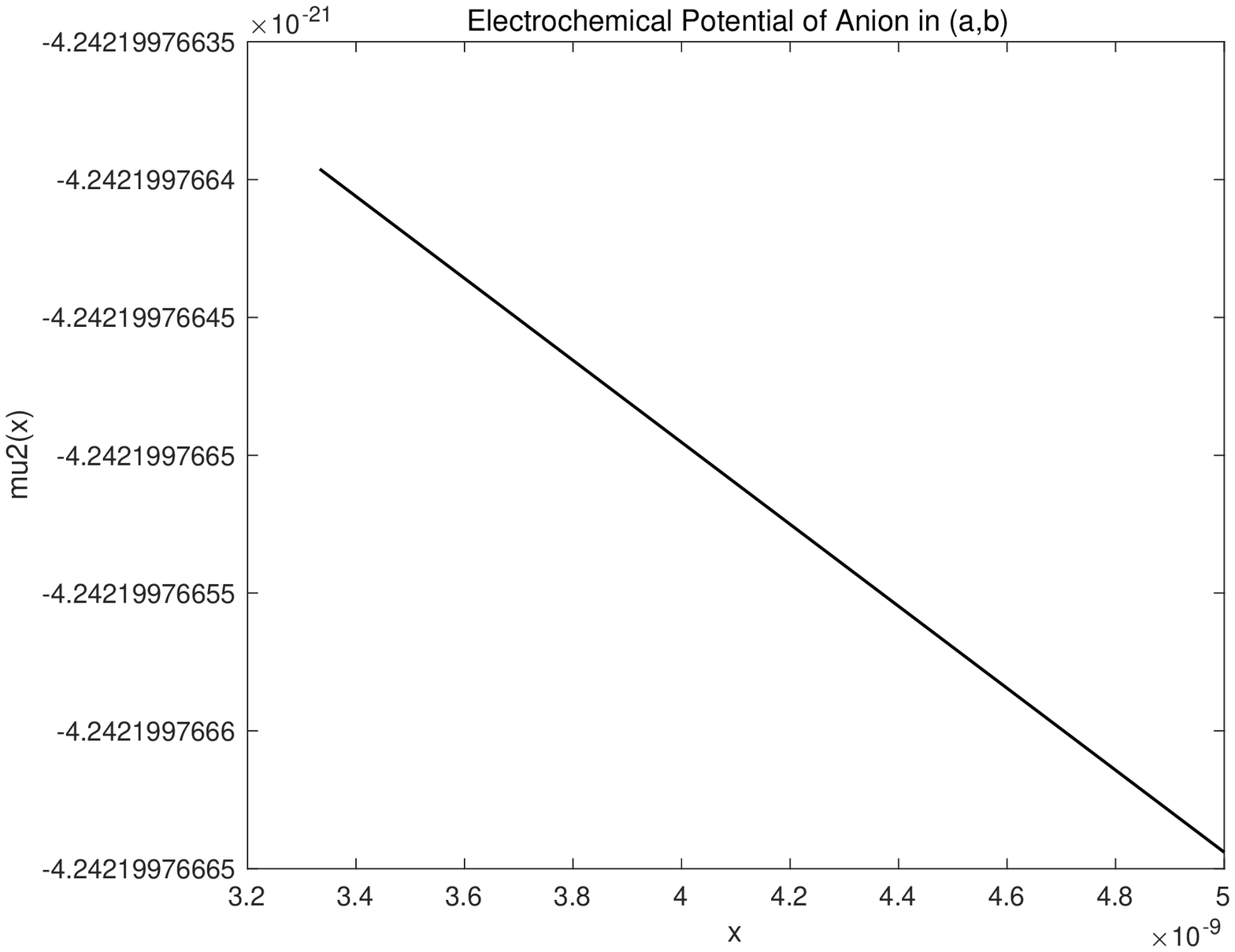}
\end{minipage}
\caption{Profiles of  $c_2(x;\nu)$,  $\phi(x;\nu)$, and $\mu_2(x;\nu)$ over interval $[a,b]$}
\end{figure} 

Note that, over this interval, $c_2(x;\nu)$ is {\em singular} in $\nu$. To understand what happens to the internal dynamics for anions over the interval $(a,b)$, we need to resolve this singularity by considering at least $O(\nu)$-terms.

 \begin{cor}\label{claim22} Over the interval $(a,b)$, 
   $c_2(x;\nu)=O(1/ \nu)\gg 1$ and $\; \mu_{2}(x;\nu)=O(\ln R)\gg 1$  BUT 
\[\mu_2'(x;\nu)\approx \mu_{21}'(x)\nu=-\frac{J_{20}}{2h(x)}\nu =O(\nu\sqrt{R}).\]
  Therefore,   from (\ref{approNP}),  
     \begin{align*} 
     J_2\approx O(1/\nu)O(\nu\sqrt{R})=O(\sqrt{R})\;\mbox{ and }\; \mu_2(x;\nu) \;\mbox{ only drops  }\; O(\nu\sqrt{R})\; \mbox{ over }\; (a,b).
      \end{align*}
         \end{cor}
    
\begin{rem}\label{lastdrop} Note that $\mu_2(x)$ drops much less over the interval $(a,b)$ than its drop over the interval $(0,a)$. But both drops are small and contribute nearly nothing to the total drop $\mu_2(0)-\mu_2(1)=O(-\ln R)\gg 1$. The only way to realize the large total drop   
 $\mu_2(0)-\mu_2(1)$ is that   $\mu_2(x)$ drops A LOT over the subinterval $(b,1)$. Indeed, this is the case as claimed below. \qed
    \end{rem}
 
  \subsubsection{Internal dynamics over the interval $(b,1)$}
   \begin{prop}\label{expProfile32} For $x\in (b,1)$,
 \begin{align*}
 \phi(x;\nu)=& \phi_0(x)+ \phi_1(x)\nu+O(\nu^2)\;\mbox{ with }\\
 \phi_0(x)=&\ln R-\ln c_{20}(x),\\
 \phi_1(x)=&\phi_1^b-\frac{1}{2}\left(B_0^2-\frac{2B_1-B_0^2}{B_0}\right)-\frac{B_0c_{21}(x)-B_1c_{20}(x)}{B_0c_{20}(x)}+\frac{2J_{11}}{J_{20}}\ln{\frac{c_{20}(x)}{B_0}};\\
 c_{20}(x)=&B_0-\frac{J_{20}}{2}(H(x)-H(b)),\quad
 c_{21}(x)=B_1-\frac{J_{11}+J_{21}}{2}(H(x)-H(b));\\
 \mu_{20}(x)=&2\ln{c_{20}(x)}-\ln R,\quad
 \mu_{21}(x)=\frac{c_{21}(x)}{c_{20}(x)}-\phi_1(x);
 \end{align*}
 where
 \begin{align*}
\phi_1^{b}=&\frac{\ln{\frac{B_0}{R}}}{\ln{\frac{B_0L}{A_0R}}}\left(2\Big(\frac{B_1}{B_0}-\frac{A_1}{A_0}\Big)+\frac{\beta-\alpha}{\alpha}(L-A_0)\right)-2\frac{B_1}{B_0}+B_0.
 \end{align*}
 \end{prop}
 
 The profiles of $c_2(x;\nu)$,  $\phi(x;\nu)$, and $\mu_2(x;\nu)$ over $(b,1)$ are shown in Figure 10.
   \begin{figure}[htbp]\label{c2three}
\centering
  \begin{minipage}[t]{0.4\linewidth}
 \includegraphics[width=2.5in]{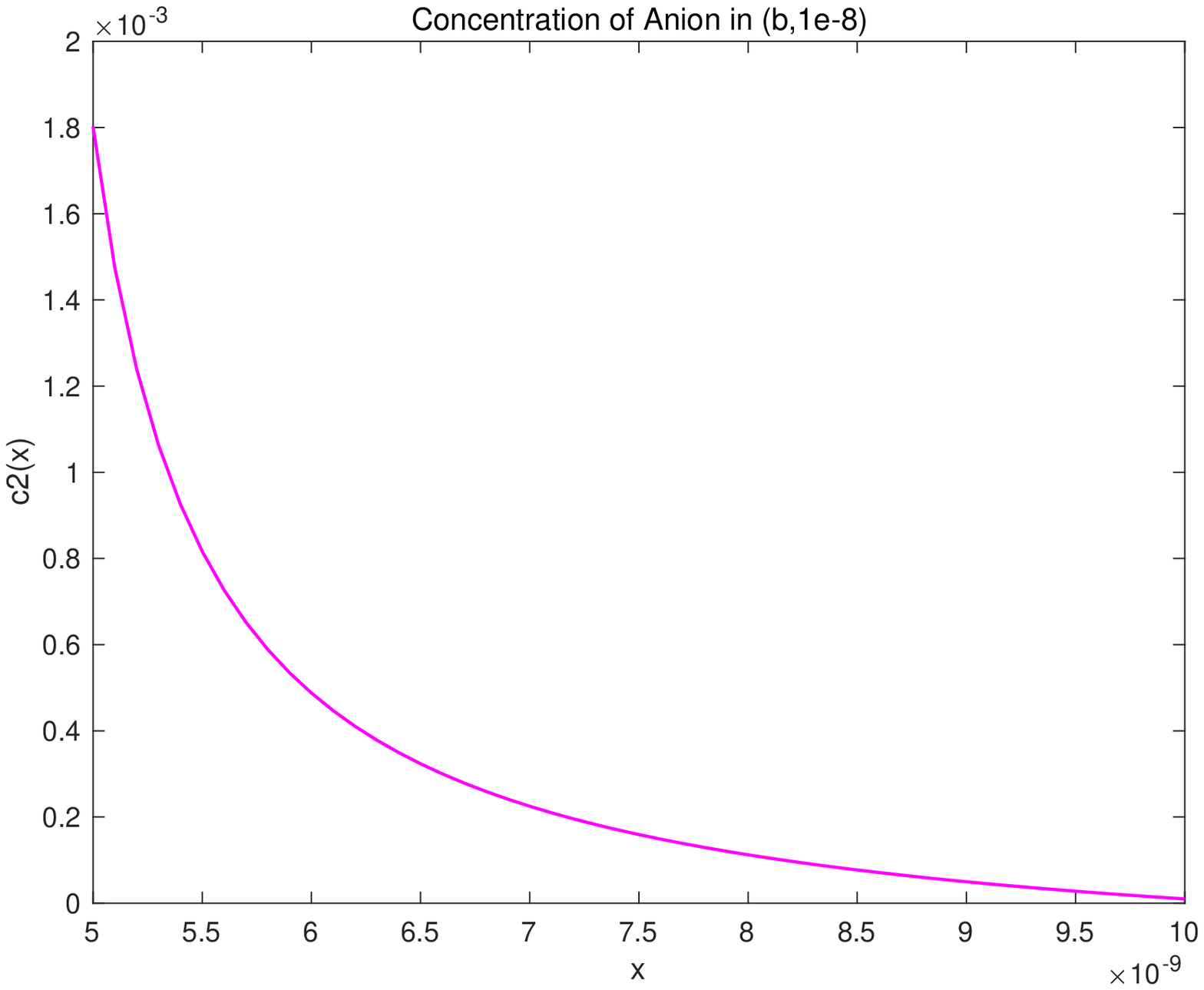}
\end{minipage}
\qquad
 \begin{minipage}[t]{0.4\linewidth}
 \includegraphics[width=2.5in]{phi_b1.eps}
\end{minipage}
 \qquad
 \begin{minipage}[t]{0.4\linewidth}
 \includegraphics[width=2.5in]{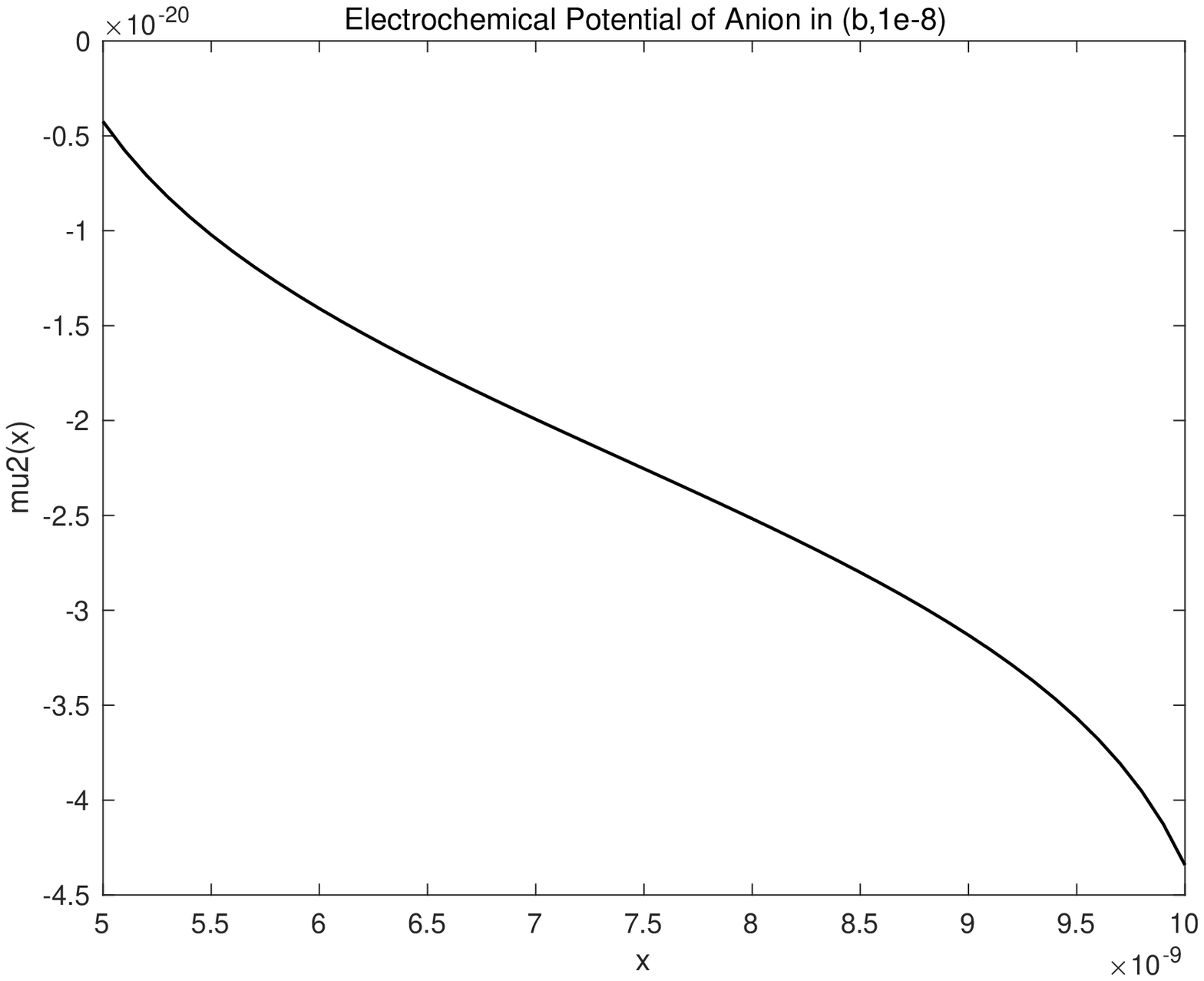}
\end{minipage}
\caption{Profiles of  $c_2(x;\nu)$, $\phi(x;\nu)$, and $\mu_2(x;\nu)$ over interval $[b,1]$}
\end{figure}

 \begin{cor}\label{claim32} Over the interval $(b,1)$,   noting $B_0=O(\sqrt{R})$,
  $c_{20}(x)$ changes from $c_{20}(b)=B_0=O(\sqrt{R})$ to $c_{20}(1)=R$ monotonically, and $\mu_2(x)$ changes from $\mu_2(b)=O(1)$ to $\mu_2(1)=\ln R$.
  
  Therefore, for $x\in (b,1)$, from (\ref{approNP}),   
     \begin{align*} 
     J_2\approx O (\sqrt{R}) \;\mbox{ and }\; \mu_2(x) \;\mbox{ drops }\; O(-\ln R)\; \mbox{ over }\; (b,1).
     \end{align*}
          \end{cor}
          
 \begin{rem}\label{difflast}
    Note that, over this interval, the order of $\mu'_2(x)$ varies in $x$ from $\mu'_2(b)=O(1)$ to $\mu'_2(1)=O(1/{\sqrt{R}})$ but overall drops  is $O(-\ln R)$. This is different from what happened over the intervals $(0,a)$ and $(a,b)$. \qed
  \end{rem}

     \subsection{Summary of mechanism for declining phenomenon.}

In summary,  with the technical assumption that $\nu\ln R\le \sqrt{R}$, we have,   as $R\to 0$,   $J_2=O(\sqrt{R})$   over the whole interval $(0,1)$   but with completely DIFFERENT scenarios over different subintervals $(0,a)$, $(a,b)$ and $(b, 1)$. More precisely,  
  \begin{itemize}
  \item[(i)]  over the subinterval $(0,a)$, one has $c_2(x)=O(1)$ but $\mu_2'(x)=O(\sqrt{R})\ll 1$ so that $J_2=O(\sqrt{R})$; ({\em Note that   the drop of $\mu_2$ over the interval $(0,a)$ is of order $O(\sqrt{R})$, which has nearly no contribution to the drop of $\mu_2$ over the whole interval $(0,1)$.})
    \item[(ii)] over the subinterval $(a,b)$, $c_2(x)=O(1/{\nu})$ but $\mu_2'(x)=O(\nu\sqrt{R})\ll 1$
      so that $J_2=O(\sqrt{R})$; ({\em Note that  the drop of $\mu_2$ over the interval $(a,b)$ is of order $O(\nu\sqrt{R})$, which is even smaller than that over the subinterval $(0,a)$ and, of course, has nearly no contribution to the drop of $\mu_2$ over the whole interval $(0,1)$.})
\item[(iii)] over the subinterval $(b,1)$,  different from what happened over each of the previous two subintervals, the orders of $c_2(x)$ and $\mu_2'(x)$ are NOT uniform in $x\in (b,1)$ but 
the drop of $\mu_2(x)$ of $O(\ln R)$ is fully realized over this subinterval (see Remark \ref{difflast}).
\end{itemize}

Here we provide the profiles of concentration (Fig.~11), electric field (Fig.~12) and electrochemical potential (Fig.~13) of the anion over the whole interval $[0,1]$.
 The figures of $c_k(x)$ and $\phi(x)$ over interval $(0,x_0)$ are not continue, because we make the plots of system \eqref{ssPNP} with $\varepsilon=0$. For the limiting system at $x=a$ and $x=b$, there are two fast layers, $c_k(x)$ and $\phi(x)$ changes very fast, but $\mu_k(x)$ keeps the same value in fast layers.
  Recall that $\mu_2^\delta=\mu_2(0)-\mu_2(1)=k_B T(-V+\ln{L}-\ln{R})$, if $R=10^{-5}$,  then $\mu_2^\delta \sim 2.4780\times 10^{-31}$, that's why in the last figure, $ \mu_2(x)$ can not reach the infinity.

   \begin{figure}[htbp]\label{c2}
\centering
  \begin{minipage}[t]{0.45\linewidth}
 \includegraphics[width=3in]{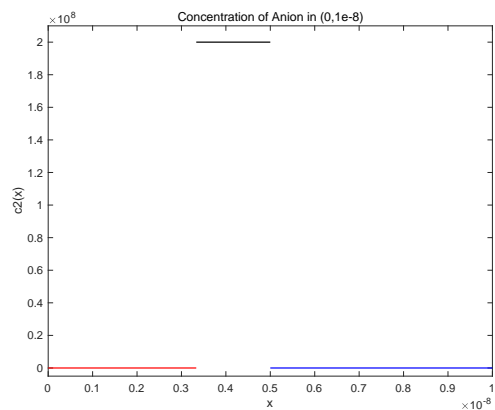}
\end{minipage}
\caption{Profile  $c_2(x;\nu)$  over $[0,1]$}
\end{figure}

\begin{figure}[htbp]\label{phi}
\centering
 \begin{minipage}[t]{0.45\linewidth}
 \includegraphics[width=3in]{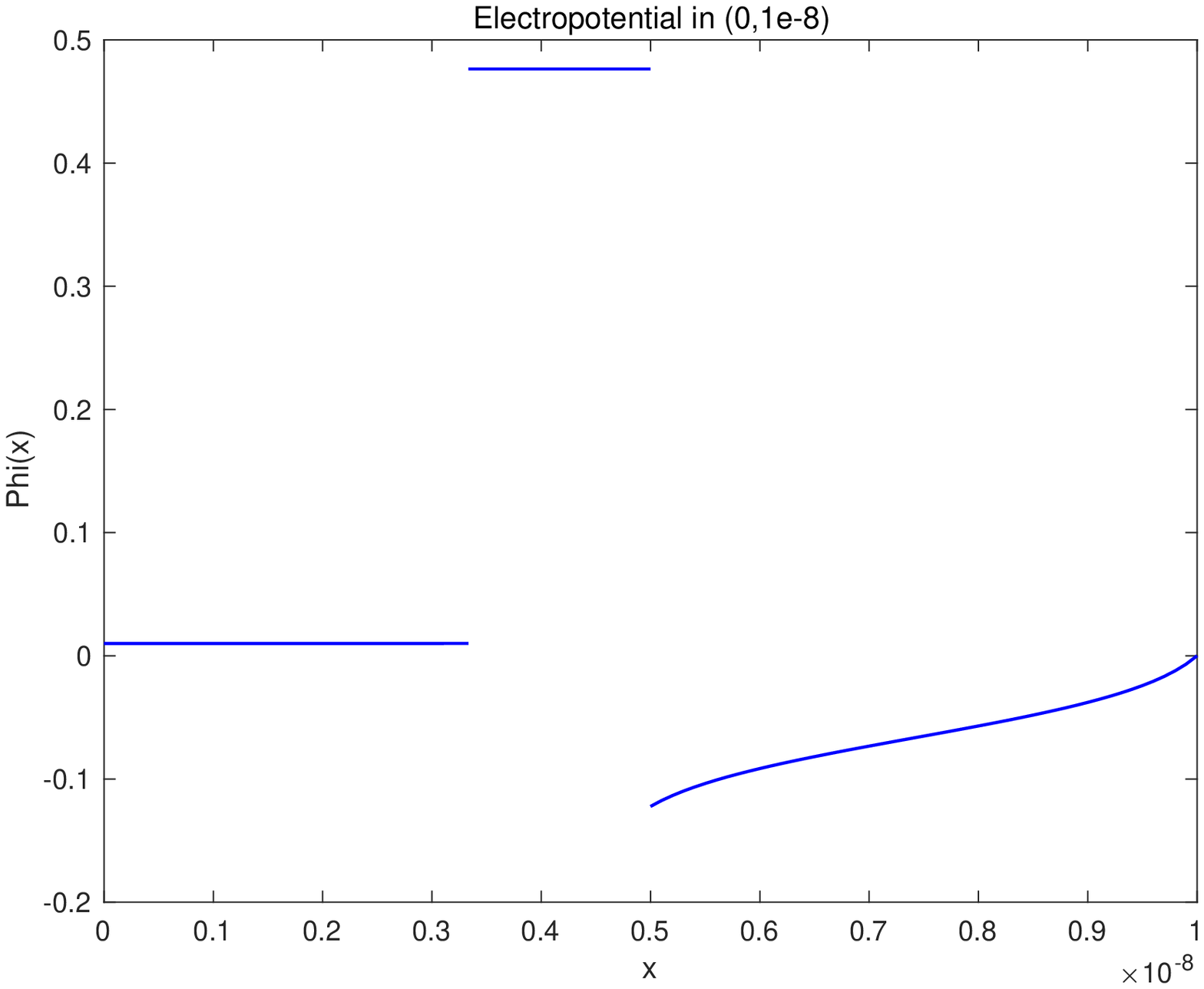}
\end{minipage}
\caption{Profile of $\phi(x;\nu)$ over $[0,1]$ }
\end{figure}
 \begin{figure}[htbp]\label{mu2}
\centering
  \begin{minipage}[t]{0.45\linewidth}
 \includegraphics[width=3in]{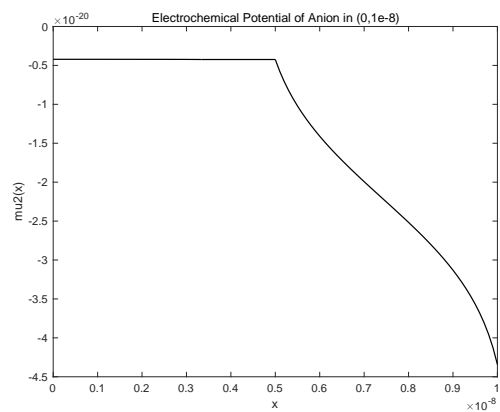}
\end{minipage}
\caption{Profile of $\mu_2(x;\nu)$ over   $[0,1]$}
\end{figure}

\newpage
\section{Concluding remarks}\label{conclude}
In this work, we examine effects of large permanent charges on ionic flow through ion channels  based on a 
quasi-one dimensional Poisson-Nernst-Planck model.  
We show that one of the defining properties of transporters, obligatory exchange, can arise in an open channel with just one structure. When the permanent charge is large, the current carried by  counter ions, majority charge carriers with the opposite sign as the permanent charge, can decline, even to zero, as the driving force (the gradient of electrochemical potential) increases. We also show that   large permanent charges essentially inhibit the flux of co-ions, regardless of the magnitude of transmembrane electrochemical potential.

  \medskip

\noindent
{\bf Acknowledgements.} It is a pleasure to thank Mordy Blaustein, Don Hilgemann and Ernie Wright for help with
 the literature of transporters and Chris Miller for help with the original formulation of the declining phenomenon  in Section \ref{Dyn4declin}.
 LZ thanks the University of Kansas for its hospitality during her visit from Oct. 2016-Oct. 2017 when this research is conducted. 
 LZ is partially supported by NNSF of China grants no. 11431008 and no. 11771282, and the Joint Ph.D. Training Program sponsored by the Chinese Scholarship Council.
   
 \bibliographystyle{plain}

\end{document}